\documentclass[twoside]{article}
\usepackage[a4paper]{geometry}
\usepackage[latin1]{inputenc} % ou \usepackage[utf8]{inputenc}
\usepackage[T1]{fontenc} % ou \usepackage[OT1]{fontenc}
\usepackage{RR}
\usepackage{hyperref}

\usepackage{amssymb}
\usepackage{amsmath,amsthm}
\usepackage{url}
\usepackage{stmaryrd}
\usepackage{float} \usepackage{pgf} \usepackage{subfig} \usepackage{multirow}
\usepackage{xspace} \usepackage{hyperref} \usepackage{algorithm,algorithmic}
\usepackage{amsmath} 
\usepackage{calc} %For \widthof

\usepackage{color}
\usepackage{bbm}

\newtheorem{theorem}{Theorem}
\newtheorem{definition}[theorem]{Definition}
\newtheorem{lemma}[theorem]{Lemma}
\newtheorem{lemma*}{Lemma}
{}
{}

\newtheorem{proposition}[theorem]{Proposition}

\newtheorem{remark}[theorem]{Remark}

\newcommand {\C}   {\mathbb C}

\newcommand {\D}   {\mathbb D}
\newcommand {\F}   {\mathbb F}
\newcommand {\Z}   {\mathbb Z}

\newcommand {\Q}   {\mathbb Q}

\newcommand{\lcm}{\tmop{lcm}}

\newcommand{\OO}{\ensuremath{{{O}}}}

\newcommand{\sO}{\ensuremath{\widetilde{{O}}}}
\newcommand{\sOB}{\ensuremath{\widetilde{{O}}_B}}

\newcommand{\tmop}[1]{\ensuremath{\operatorname{#1}}}

\newcommand{\makeremark}[2]{
  \newcommand{#1}[1]
    {
    \color{blue}
     $\longrightarrow$ \textsc{#2: }
     ##1
     $\longleftarrow$
    \color{black}
    }
}    

\makeremark{\SL}{SL}
\makeremark{\FR}{Fabrice says}
\makeremark{\MP}{Marc says}
\makeremark{\ET}{Elias says}
\makeremark{\LP}{Luis shouts}
\makeremark{\YB}{Yacine says}

\definecolor{1ST}{rgb}{1,0,0}%red
\definecolor{2ND}{rgb}{1,0.5,0}%orange
\definecolor{3RD}{rgb}{1,0,1}%pink

\newcommand{\ideal}[1]{\langle #1 \rangle }

\newcommand{\shutup}[1]{}

\renewcommand{\leq}{\leqslant}  %needs \usepackage{amssymb}
\renewcommand{\geq}{\geqslant}

\makeatletter
\def\cramped                           % "Cramped" list style.
 {\parskip0pt\@topsep0pt       % put after \begin{itemize}
  \itemsep0pt\parsep0pt
}        
\makeatother

\usepackage{stmaryrd}

\usepackage{verbatim}
\newcommand {\comp}{\ensuremath{\sO_B(d^8+d^7\tau)}} %+d^5\tau ^2

\RRNo{8262}
\RRdate{March 2013}
\RRauthor{
Yacine Bouzidi\thanks{INRIA Nancy Grand Est, LORIA laboratory, Nancy, France. {\tt Firstname.Name@inria.fr}}
        \and
        Sylvain Lazard\footnotemark[1] 
        \and
        Marc Pouget\footnotemark[1] 
        \and
        Fabrice Rouillier\thanks{INRIA Paris-Rocquencourt and IMJ
           (Institut de Math\'ematiques de Jussieu, Universit\'e Paris 6, CNRS), Paris, France. {\tt Firstname.Name@inria.fr}}
}
\authorhead{Bouzidi \& Lazard \& Pouget \& Rouillier}
\RRtitle{Représentations  Univariées Rationnelles de systèmes bivariés et
  applications}
\RRetitle{Rational Univariate Representations of bivariate systems and applications}
\titlehead{Rational Univariate Representations of Bivariate Systems and Applications}
\RRresume{Nous abordons le probl\`eme de la r\'esolution de syst\`emes de deux
  polyn\^omes \`a deux variables de  degr\'e total au plus $d $ \`a coefficients
  entiers de bitsize maximale  $ \tau $. Il est connu que une forme
  lin\'eaire s\'eparante, c'est-\`a-dire une combinaison lin\'eaire des
  variables qui prend des valeurs diff\'erentes quand elle est \'evalu\'ee en
  des solutions (complexes) distinctes du syst\`eme, peut \^etre calcul\'ee en
  \comp\ bits  op\'erations (o\`u $ \sO $ se r\'ef\`ere \`a la complexit\'e
  o\`u les facteurs polylogarithmiques sont omis et $O_B$ se r\'ef\`ere \`a la
  complexit\'e binaire) et nous nous concentrons ici sur le  calcul d'une
  repr\'esentation univari\'ee rationnelle (RUR) \'etant donn\'e une forme
  lin\'eaire s\'eparante.

  Nous pr\'esentons un algorithme pour le calcul d'une RUR de complexit\'e $ \sOB
  (d^7 + d^6 \tau) $ dans le pire cas et nous bornons la taille de ses
  coefficients par $ \sO(d^2 + d\tau) $.  Nous montrons en outre que des bo\^ites
  d'isolation des solutions du syst\`eme peuvent \^etre calcul\'ees \`a partir de la RUR
  avec \comp\ bits op\'erations. Enfin, nous montrons comment une RUR peut \^etre
  utilis\'ee pour \'evaluer le signe d'un polyn\^ome \`a deux variables (de degr\'e au plus
  $ d $ et bitsize au plus $ \tau $) en une solution r\'eelle du syst\`eme en
  \comp\ bits op\'erations et en toutes les $ \Theta(d^2) $  solutions
  r\'eelles en  seulement $ O(d) $ fois plus que pour une solution.  }

 \RRabstract{ We address the problem of solving systems of two bivariate polynomials of
  total degree at most $d$ with integer coefficients of maximum bitsize
  $\tau$. {It is known that a linear separating form, that is a linear
    combination of the variables that takes different values at distinct
    solutions of the system, can be computed in $\sOB(d^{8}+d^7\tau)$ bit
    operations (where $O_B$ refers to bit complexities and $\sO$ to complexities
    where polylogarithmic factors are omitted) and we focus here on the
    computation of a Rational Univariate Representation (RUR) given a linear
    separating form.}
We present an
  algorithm for computing a RUR with worst-case bit complexity in
  $\sOB(d^7+d^6\tau)$ and bound the bitsize of its coefficients by $\sO(d^2+d\tau)$.
  We show in addition that isolating boxes of the solutions of the system can be
  computed from the RUR with $\sOB(d^{8}+d^7\tau)$ bit operations.  Finally, we
  show how a RUR can be used to evaluate the sign of a bivariate polynomial (of
  degree at most $d$ and bitsize at most $\tau$) at one real solution of the
  system in $\sOB(d^{8}+d^7\tau)$ bit operations and at all the $\Theta(d^2)$
  {real} solutions in only $O(d)$ times that for one solution.}
\RRmotcle{calcul formel, r\'esolution de syst\`emes polynomiaux,
  repr\'esentation univari\'ee rationnelle
} \RRkeyword{computer algebra, polynomial system solving, Rational Univariate
  Representations}
 \RRprojet{Vegas}  % cas d'un seul projet
 \URLorraine % pour ceux qui sont \`a l'est
 \RCNancy % Nancy - Grand Est
\begin{document}
\makeRR  

\section{Introduction}

There exists many algorithms, in the literature, for ``solving'' algebraic
systems of equations.  Some focus on computing ``\emph{formal solutions}'' such
as rational parameterizations, Gr\"obner bases, and triangular sets, others
focus on isolating the solutions.  By isolating the solution, we mean computing
isolating axis-parallel boxes sets such that every real solution lies in a
unique box and conversely.  In this paper, we focus on the worst-case bit
complexity of these methods (in the RAM model) for systems of {\bf bivariate
  polynomials of total degree $\mathbf{d}$ with integer coefficients of bitsize
  $\mathbf{\tau}$.}

It should be stressed that formal solutions do not necessarily yield, directly,
isolating boxes of the solutions. In particular, from a theoretical complexity
view, it is not proved that the knowledge of a triangular system or Gr\"obner
basis of a system always simplifies the isolation of its solutions. The
difficulty lies in the fact that isolating the solutions of a triangular system
essentially amounts to isolating the roots of univariate polynomials with
algebraic numbers as coefficients, which is not trivial when these polynomials
have multiple roots. For recent work on this problem, we refer to
\cite{cgy-issac-2007,BCLM} where no upper bound of complexity are given for the
roots isolation.  This difficulty also explains why it is not an easy task to
define precisely what a formal solution of a system is, and why usage prevails
in what is usually considered to be a formal solution.

For isolating the real solutions of systems of two bivariate polynomials, the
algorithm with best known bit complexity was recently analyzed by Emeliyanenko
and Sagraloff \cite{sagraloff2012issacBisolve}. They solve the problem in
$\sOB(d^8+d^7\tau)$ bit operations (where $\sO$ refers to the complexity where
polylogarithmic factors are omitted and $O_B$ refers to the bit complexity).
Furthermore, the isolating boxes can easily be refined because the algorithm
computes the univariate polynomials %, in each variable $x$ and $y$, 
that
correspond to the projections of the solutions on each axis (that is, the
resultants of the two input polynomials with respect to each of the variables).
% \SL{Vire la ref a Pan vu que $\sOB(d^8+d^7\tau)$ ne match plus l'isolation du
% sqrfree du resultant} Also, the complexity of their algorithm is remarkable
% because it matches the best known complexity of {numerical} algorithms for
% isolating the roots of the squarefree part of one of these resultants (see
% \cite{Pan97,Schonhage});
%%Also, the complexity of their algorithm is likely to
%%be not far from the optimum because their complexity matches the best known
%%complexity for isolating the roots of these resultants; 
%thus improving their complexity would require to either prove that isolating the
%solutions of a system is ``simpler'' than isolating the roots of one of its
%resultants (which is not very likely), or to improve the complexity of isolating
%the roots of these resultants (see \cite{sagraloff2012issacBisolve} for
%details).
 The main drawback of their approach is, however, that their
output %solution
(i.e., the isolating boxes and the two resultants) does not seem to help for
performing some important operations on the solutions of the system, such as
computing the sign of a polynomial at one of these real  solutions (referred to as the
\emph{sign\_at} operation), which is a critical operation in many problems, in
particular in geometry.

Other widespread approaches that solve systems 
and
allow for simple \emph{sign\_at} evaluations, are those that consist in computing
rational parameterizations of the (complex) solutions.  Recall that such
  a rational parameterization is a set of univariate polynomials and associated
  rational one-to-one mappings that send the roots of the univariate polynomials
  to the solutions of the system.
  The algorithm with the best known complexity for solving such systems via
  rational parameterizations was, in essence, first introduced by Gonzalez-Vega
  and El Kahoui \cite{VegKah:curve2d:96} (see also
  \cite{lgv-in-etdidapc-02}). The algorithm first applies a generic linear
  change of variables to the two input polynomials, computes a rational
  parameterization using the subresultant sequence of the sheared polynomials and
  finally computes the isolating boxes of the solutions. Its initial bit
  complexity of $\sOB(d^{16}+d^{14}\tau ^2)$ was improved by Diochnos et al.
  \cite[Theorem 19]{det-jsc-2009} to (i) $\sOB(d^{10}+d^9\tau)$ for computing a
  generic shear (i.e., a separating linear form), to (ii) $\sOB(d^{7}+d^6\tau)$
  for computing a rational parameterization and to (iii) $\sOB(d^{10}+d^9\tau)$
  for the isolation phase with a modification of the initial
  algorithm.\footnote{The complexity of the isolation phase in \cite[Theorem
    19]{det-jsc-2009} is stated as $\sOB(d^{12}+d^{10}{\tau}^2)$ but it
    trivially decreases to $\sOB(d^{10}+d^9\tau)$ with the recent result of
    Sagraloff \cite{sagraloff2012issacNewDsc} which improves the complexity of
    isolating the real roots of a univariate polynomial. Note also that Diochnos
    et al.  \cite{det-jsc-2009} present two algorithms, the M\_RUR and G\_RUR
    algorithms, both with bit complexity $\sOB(d^{12}+d^{10}{\tau}^2)$. However,
    this complexity is worst case only for the M\_RUR algorithm. As pointed out
    by Emeliyanenko and Sagraloff \cite{sagraloff2012issacBisolve}, the G\_RUR
    algorithm uses a modular gcd algorithm over an extension field whose
    considered bit complexity is expected.}

  \paragraph{Main results.} We addressed in \cite{bouzidi2013SepElt} the first phase of the above
  algorithm and proved that, given two polynomials $P$ and $Q$ of degree at most $d$ and bitsize at
  most $\tau$, a separating linear form can be computed in $\sOB(d^{8}+d^7\tau)$ bit operations
  (improving by a factor $d^2$ the above complexity).  We suppose computed such a separating linear
    form and address in this paper the second and third phase of the above algorithm, that is the
  computation of a rational parameterization and the isolation of the solutions of the system.  We
  also consider two important related problems, namely, the evaluation of the sign of a polynomial
  at the real solutions of a system and the computation of a rational parameterization of
  over-constrained systems.

  We first show that the Rational Univariate Representation (RUR for short) of Rouillier
  \cite{Rou99} (i) can be expressed with simple polynomial formulas, that (ii) it has a total
  bitsize which is asymptotically smaller than that of Gonzalez-Vega and El Kahoui by a factor $d$,
  and that (iii) it can be computed with the same complexity, that is $\sOB(d^{7}+d^6\tau)$
  (Theorem~\ref{th:rur}). Namely, we prove that the RUR consists of four polynomials of degree at
  most $d^2$ and bitsize $\sO(d^2+d\tau)$ (instead of $O(d)$ polynomials with the same asymptotic
  degree and bitsize for Gonzalez-Vega and El Kahoui parameterization).  Moreover, we prove that
  this bound holds for any ideal containing $P$ and $Q$, that is, for instance the radical ideal of
  $\ideal{P,Q}$ (Proposition~\ref{prop:rur-size}).

We show that, given a RUR, isolating boxes of the solutions of the system can be computed with
$\sOB(d^{8}+d^7\tau)$ bit operations (Proposition~\ref{prop:computing-boxes}).  This decreases by a
factor $d^2$ the best known complexity for this isolation phase of the algorithm (see the discussion
above).  Globally, this brings the overall bit complexity of all three phases of the algorithm to
$\sOB(d^{8}+d^7\tau)$, which also improves by a factor $d^2$ the complexity.

Finally, we show how a rational parameterization can be used to perform efficiently two important
operations on the input system. We first show how a RUR can be used to perform efficiently the
\emph{sign\_at} operation.  Given a polynomial $F$ of total degree at most $d$ with integer
coefficients of bitsize at most $\tau$, we show that the sign of $F$ at one real solution of the
system can be computed in $\sOB(d^8+d^7\tau)$ bit operations, while the complexity of computing its
sign at all the $\Theta(d^2)$ solutions of the system is only $O(d)$ times that for one real
solution (Theorem~\ref{th:sign_at_rur}).  This improves the best known complexities of
$\sOB(d^{10}+d^9\tau)$ and $\sOB(d^{12}+d^{11}\tau)$ for these respective problems (see \cite[Th. 14
\& Cor. 24]{det-jsc-2009} with the improvement of \cite{sagraloff2012issacNewDsc} for the root
isolation).
Similar to the \emph{sign\_at} operation, we show that a RUR can be split in two parameterizations
such that $F$ vanishes at all the solutions of one of them and at none of the other. {We also show
  that these rational parameterizations can be transformed back into RURs in order to reduce their
  total bitsize (see above), within the same complexity, that is, $\sOB(d^8+d^7\tau)$}
(Proposition~\ref{prop:overconstrained}).

\medskip The paper is organized as follows: in Section~\ref{sec:RURalgo}, we present our algorithm
for computing the RUR based on the formulas of Proposition~\ref{prop:rur-res2}.  We then use these
formulas in Section~\ref{sec:size-rur} to prove new bounds on the bitsize of the coefficients of the
polynomials of the RUR.  {The main results of Section~\ref{sec:rur-candidate} are summarized in
  Theorem~\ref{th:rur}.}  In Section~\ref{sec:applications}, we present three applications of the
RUR. We first describe in Section~\ref{sec:boxes} an algorithm for isolating the real solutions.  We
then present in Section~\ref{sec:sign_at} an algorithm for computing the sign of a bivariate
polynomial at these solutions and, finally, we show in Section~\ref{sec:overcontraint} how a RUR can
be split into rational parameterizations whose solutions satisfy some equality and inequality
constraints.

\section{Notation and preliminaries}
%%%%%%%%%%%%%%%%%%%%%%%%%%%%%%%%%%%%%%%%%%%%%%%%%%%%%%%%%

We introduce notation and recall the definition of subresultant sequences and
basics of complexity.

The bitsize of an integer $p$ is the number of bits needed to represent it, that
is $\lfloor\log p\rfloor+1$ ($\log$ refer to the logarithm in base 2). For
rational numbers, we refer to the bitsize as to the maximum bitsize of its
numerator and denominator.
The bitsize of a polynomial with integer or rational coefficients is the
\emph{maximum}  bitsize of its coefficients. 
We refer to $\tau_\gamma$ as the bitsize of a polynomial, rational or integer $\gamma$.

We denote by $\D$ a unique factorization domain,
typically $\Z[X,Y]$, $\Z[X]$ or $\Z$. We also denote by
$\F$ a field, typically $\Q$, $\C$.
For any polynomial $P\in \D[X]$, let $Lc_X(P)$ denote its leading coefficient
with respect to the variable $X$ (or simply $Lc(P)$ in the univariate case),
$d_X(P)$ its degree with respect to $X$, and $\overline{P}$ its squarefree
part. The ideal generated by two polynomials $P$ and $Q$ is denoted
$\ideal{P,Q}$, and the affine variety of an ideal $I$ is denoted by $V(I)$; in
other words, $V(I)$ is the set of distinct solutions of the system
$\{P,Q\}$. The solutions are always considered in the algebraic
closure of the fraction field of $\D$,
unless specified otherwise. For a point $\sigma \in V(I)$, $\mu_I(\sigma)$
denotes the multiplicity of $\sigma$ in $I$. For simplicity, we refer
indifferently to the ideal $\ideal{P,Q}$ and to the corresponding system of
polynomials.

\emph{We finally introduce the following notation which are extensively used
  throughout the paper.} Given the two input polynomials $P$ and $Q$, we
consider the ``generic'' change of variables $X=T-SY$, and define the
``sheared'' polynomials $P(T-SY,Y)$, $Q(T-SY,Y)$, and their resultant with
respect to $Y$,
\begin{equation}\label{eq0}
{ R(T,S)}=Res_Y({P}(T-SY,Y),{Q}(T-SY,Y)).
\end{equation}
Let $L_R(S)$ be the leading coefficient of $R(T,S)$ seen as a polynomial in $T$.
Let  ${ L_{P}(S)}$ and ${ L_{Q}(S)}$ be the leading coefficients
 of ${P}(T-SY,Y)$ and $Q(T-SY,Y)$, seen as polynomials in $Y$;  it is straightforward   that these leading coefficients do
 not depend on $T$. In other words: 
%\[L_{P}(S) = Lc_Y({P}(T-SY,Y)),\quad\quad L_{Q}(S) = Lc_Y({Q}(T-SY,Y)),\]
%\[L_{PQ}(S) =L_{P}(S) L_{Q}(S),\quad\quad L_R(S)= Lc_T(R(T,S)).\]
%

\begin{equation}\label{eq1}
\begin{array}{c}
L_{P}(S) = Lc_Y({P}(T-SY,Y)),\ \ L_{Q}(S) = Lc_Y({Q}(T-SY,Y)) \\
 L_R(S)= Lc_T(R(T,S)) %% \\
%% L_{PQ}(S) =L_{P}(S) L_{Q}(S),
\end{array}
\end{equation}

\paragraph{Complexity.} 
We recall some complexity bounds.
In the sequel, we often consider the gcd of two univariate polynomials $P$ and $Q$ and the gcd-free part of $P$ with respect to $Q$, that is, the divisor $D$ of $P$ such that
$P=\gcd(P,Q)D$. Note that when $Q=P'$, $D$ is the squarefree part $\overline{P}$ of~$P$.

\begin{lemma}[{\cite[Corollary 10.12 \& Remark 10.19]{BPR06}}\footnote{{\cite[Corollary 10.12]{BPR06} states that  $P$ and $Q$ have a gcd in $\Z[X]$
    with bitsize in $O(d+\tau)$.  \cite[Remark 10.19]{BPR06} claims that a gcd and gcd-free parts
    of $P$ and $Q$ can be computed in  $\sOB(d^2\tau)$ bit
    operations. This remark refers to \cite[Corollary
    5.2]{lr-jsc-2001} which proves that the last non-zero Sylvester-Habicht polynomial, which is a
    gcd of $P$ and $Q$ \cite[Corollary 8.32]{BPR06}, can be computed in $\sOB(d^2\tau)$ bit
    operations. Moreover, the corollary proves that the Sylvester-Habicht transition matrices can be
    computed within the same bit complexity, which gives the cofactors of $P$ and $Q$ in the
    sequence of the Sylvester-Habicht polynomials (i.e., $U_i, V_i\in \Z[X]$ such that $U_iP+V_iQ$
    is equal to the $i$-th Sylvester-Habicht polynomials).  The gcd-free part of $P$ with respect to $Q$
    and conversely are the cofactors corresponding to the one-after-last non-zero Sylvester-Habicht
    polynomial \cite[Proposition 10.14]{BPR06}, and can thus be computed in  $\sOB(d^2\tau)$ bit
    operations.
    The gcd (resp.  gcd-free part) of $P$ and $Q$ computed this way is in
    $\Z[X]$, thus dividing it by the gcd of its coefficients yields a gcd (resp.
    gcd-free part) of $P$ and $Q$ of smallest bitsize in $\Z[X]$ which is known
    to be in $O(d+\tau)$. The gcd of the coefficients, which are of bitsize
    $\sO(d\tau)$ \cite[Proposition 8.46]{BPR06}, follows from $O(d)$ gcds of two
    integers of bitsize $\sO(d\tau)$ and each such gcd can be computed
    with $\sOB(d\tau)$ bit operations \cite[\S 2.A.6]{Yap-2000}.  Therefore, a
    gcd (resp.  gcd-free part) of $P$ and $Q$ of bitsize $O(d+\tau)$ can be
    computed in $\sOB(d^2\tau)$ bit complexity.}}]
\label{complexity:gcd}
Two polynomials $P$, $Q$ in $\Z[X]$ with maximum degree $d$ and bitsize at
most $\tau$ have a gcd in $\Z[X]$ 
with coefficients of bitsize in $O(d+\tau)$ which
can be computed with $\sOB(d^2\tau)$ bit operations.  The same bounds hold for
the bitsize and the computation of the gcd-free part of $P$ with respect to
$Q$.  
\end{lemma}

The following is a refinement of the previous lemma for the case of two
  polynomials with different degrees and bitsizes. It is a straightforward
adaptation of \cite[Corollary 5.2]{lr-jsc-2001}
and it is only used in Section~\ref{sec:overcontraint}.

 \begin{lemma}[{\cite{lr-jsc-2001}}\footnote{The algorithm in \cite{lr-jsc-2001} uses
the well-known half-gcd approach to compute any polynomial in the Sylvester-Habicht and cofactors sequence in a
    soft-linear number of arithmetic operations, and it exploits Hadamard's bound
    on determinants to bound the size of intermediate
    coefficients. 
When the two input  polynomials have different degrees and bitsizes, Hadamard's bound reads as
$\sO(p\tau_Q+q\tau_P)$ instead of simply $\sO(d\tau)$ and, 
similarly as in Lemma~\ref{complexity:gcd}, the algorithm in \cite{lr-jsc-2001} 
yields a gcd and gcd-free parts of $P$ and $Q$ in $\sOB(\max(p,q)(p\tau_Q+q\tau_P))$ bit operations.
Furthermore,  the gcd and gcd-free parts computed this way are in $\Z[X]$ with
coefficients of bitsize $\sO(p\tau_Q+q\tau_P)$, thus, dividing them by the gcd of their coefficients
can be done with $\sOB(\max(p,q)(p\tau_Q+q\tau_P))$ bit operations and yields a  gcd and  gcd-free
parts in $\Z[X]$ with minimal bitsize, which is as claimed by Mignotte's bound
(see e.g. \cite[Corollary 10.12]{BPR06}).}]
\label{lem:finegcd}
  Let $P$ and $Q$ be two polynomials in $\Z[X]$ of degrees $p$ and $q$ and of  bitsizes $\tau_P$ and
  $\tau_Q$, respectively. 
A gcd of $P$ and $Q$ of bitsize $O(\min(p+\tau_P,q+\tau_Q))$  in $\Z[X]$, can be computed
in $\sOB(\max(p,q)(p\tau_Q+q\tau_P))$ bit operations. 
A gcd-free part of $P$ with respect to $Q$, of bitsize $O(p+\tau_P)$ in $\Z[X]$, can be computed in the same bit complexity.
 \end{lemma}

We now state a bound on the complexity of evaluating a univariate
   polynomial which ought to be known, even though we were not able to find a
   proper reference to it.  For completeness, we provide a very simple proof.

\begin{lemma}%\cite{2011bodrato}
\label{lem:comp:evaluation} 
Let $a$ be a rational of bitsize $\tau_a$, the evaluation at $a$ of a univariate 
polynomial $f$ of degree $d$ and rational coefficients of bitsize $\tau$ can be done in
$\sO_B(d(\tau+\tau_a))$ bit operations, while the value $f(a)$ has bitsize in
$O(\tau+d\tau_a)$.
\end{lemma}
\begin{proof}
  The complexity $\sO_B(d(\tau+\tau_a))$ can easily be obtained by recursively
  evaluating the polynomial $\sum_{i=0}^d a_i\, x^i$ as $\sum_{i=0}^{d/2} a_i\, x^i
  + x^{d/2}\sum_{i=1}^{d/2} a_{i+d/2}\, x^i$. Evaluating $x^{d/2}$ can be done in
  $O_B(d\tau_a\log^3 d\tau_a)$ time by recursively computing $\log \frac{d}{2}$
  multiplications of rational numbers of bitsize at most $d\tau_a$, each of
  which can be done in $O_B(d\tau_a\log d\tau_a\log\log d\tau_a)$ time by
  Sch\"onhage-Strassen algorithm (see e.g. \cite[Theorem 8.24]{vzGGer}.
  $\sum_{i=0}^{d/2} a_{i+d/2}\,a ^i$ has bitsize at most $d\tau_a+\tau$, hence its
  multiplication with $a^{d/2}$ can be done in $O_B((d\tau_a+\tau)\log^2
  (d\tau_a+\tau))$ time.  Hence, the total complexity of evaluating $f$ is at
  most $T(d,\tau,\tau_a)=2T(d/2,\tau,\tau_a)+O_B((d\tau_a+\tau)\log^3
  (d\tau_a+\tau))$ which is in\footnote{Indeed, $T(d,\tau,\tau_a)=2^{i+1}T(\frac{d}{2^{i+1}},\tau,\tau_a)+O_B(
    (d\tau_a+\tau)\log^3 (d\tau_a+\tau) + \cdots+ 2^i
    (\frac{d}{2^i}\tau_a+\tau)\log^3
    (\frac{d}{2^i}\tau_a+\tau))$

\indent
\newlength{\lengthone}
\settowidth{\lengthone}{Indeed, $T(d,\tau,\tau_a)$}
\hspace{\lengthone}
$ \leq O_B( d\tau_a\log^3 (d\tau_a+\tau) \log d + \tau\log^3 (d\tau_a+\tau)\sum_{i=0}^{\log d} 2^i)$

\indent
\newlength{\lengthtwo}
\settowidth{\lengthtwo}{Indeed, $T(d,\tau,\tau_a)$}
\hspace{\lengthtwo}
$ \leq O_B( d(\tau_a+\tau)\log^4 (d\tau_a+\tau))$.}
$O_B(d(\tau_a+\tau)\log^4 (d\tau_a+\tau))$ that is in $\sOB(d(\tau_a+\tau))$.
\end{proof}

\begin{lemma}[{\cite[Lemma 5]{bouzidi2013SepElt}}]\label{lem:complexity:shear}
  Let $P$ and $Q$ in $\Z[X,Y]$ {be} of total degree {at most} $d$ and maximum bitsize~$\tau$. 
The sheared polynomials $P(T-SY,Y)$ and $Q(T-SY,Y)$ can be
  expanded in $\sOB(d^4+d^3\tau)$ and their bitsizes are in
  $\sO(d+\tau)$. The resultant $R(T,S)$ can be computed in
  $\sOB(d^7+d^6\tau)$ bit operations and $\sO(d^5)$ arithmetic
  operations in $\Z$; its degree is at most $2d^2$ in each variable
  and its bitsize is in $\sO(d^2+d\tau)$.
\end{lemma}

\section{Rational Univariate Representation}
%%%%%%%%%%%%%%%%%%%%%%%%%%%%%%%%%%%%%%%%%%%%%%%%%
\label{sec:rur-candidate}

The idea of this section is to express the polynomials of a RUR of two
polynomials in terms of a resultant defined from these polynomials. Given a
separating form, this yields a new algorithm to compute a RUR and it also
enables us to derive the bitsize of the polynomials of a~RUR.
In Section~\ref{sec:RURalgo}, we prove these expressions for the polynomials of
a RUR and present the corresponding algorithm. We prove the bound on the bitsize
of the RUR in Section~\ref{sec:size-rur}. These results are summarized in Theorem~\ref{th:rur}.

Throughout this section we assume that the two input polynomials $P$ and
$Q$ are coprime in $\Z[X,Y]$, that their
  maximum total degree $d$ is at least 2 and that their coefficients have
  maximum bitsize $\tau$. 

We first recall the definition and main properties of Rational Univariate
Representations. In the following, for any polynomial $v\in \Q[X,Y]$ and
$\sigma=(\alpha,\beta)\in \C^2$, we denote by $v(\sigma)$ the image of $\sigma$ by the
polynomial function $v$ (e.g. $X(\alpha,\beta)=\alpha$).

\begin{definition}[\cite{Rou99}]
\label{def:rur}
Let $I\subset \Q[X,Y]$ be a zero-dimensional ideal, $V(I)=\{\sigma \in \C^2,
v(\sigma)=0,\forall v\in I \}$ its associated variety, and a linear form
$T=X+aY$ with $a\in \Q$.  The RUR-candidate of $I$ associated to $X+aY$ (or
simply, to $a$), denoted $RUR_{I,a}$, is the following set of four univariate
polynomials in~$\Q[T]$
\begin{equation}\label{eq:defRUR}
\begin{split}
&\displaystyle f_{I, a} (T) = \prod_{\sigma \in V (I)} (T - X (\sigma) - aY
    (\sigma))^{\mu_I (\sigma)} \\
&\displaystyle    f_{I, a, v} (T) = \sum_{\sigma \in V (I)} \mu_I (\sigma) v (\sigma)
    \prod_{\varsigma \in V (I), \varsigma \neq \sigma} (T - X (\varsigma) - aY
    (\varsigma)),  \quad \mbox{for } v\in \{1,X,Y \}
\end{split}
\end{equation}
where, for $\sigma \in V(I)$, $\mu_I(\sigma)$ denotes the multiplicity of
$\sigma$ in $I$.  If $(X,Y)\mapsto X+aY$ is injective on $V(I)$, we say that the
linear form $X+aY$ separates $V(I)$ (or is separating for $I$) and
$RUR_{I,a}$ is called
  a RUR (the RUR of $I$ associated to $a$) and it defines a bijection
  between $V(I)$ and $V(f_{I,a})=\{\gamma\in \C,f_{I,a}(\gamma)=0\}$:
 $$\begin{array}{ccc}
   V(I) & \rightarrow & V(f_{I,a}) \\
   (\alpha,\beta) & \mapsto & \alpha + a \beta \\
\displaystyle   \left(\frac{f_{I,a,X}}{f_{I,a,1}}(\gamma),\frac{f_{I,a,Y}}{f_{I,a,1}}(\gamma)\right) &  \mapsfrom  &
   \gamma \\ 
 \end{array}$$
Moreover, this bijection preserves the real roots and the multiplicities.
\end{definition}

We prove in this section the following theorem on the RUR of two polynomials. We state it for any separating linear form $X+aY$ with integer $a$ of bitsize $\sO(1)$
with the abuse of notation that polylogarithmic factors in $d$ and $\tau$ are omitted.
Note that it is known that there exists a separating form $X+aY$ with a positive  integer $a<2d^4$ and
that such a separating form can be computed in $\comp$ bit operations \cite{bouzidi2013SepElt}. This
theorem is  a direct consequence of Propositions~\ref{prop:complexity-rur-cand} and~\ref{prop:rur-size}.

\begin{theorem}\label{th:rur}
  Let $P, Q \in \mathbb{Z} [X, Y]$ be two coprime bivariate polynomials of total
  degree at most $d$ and maximum bitsize $\tau$. Given a separating form $X+aY$
  with integer $a$   of bitsize $\sO(1)$, 
the RUR of $\ideal{P,Q}$
  associated to $a$ can be computed using Proposition~\ref{prop:rur-res2} with $\sOB(d^7+d^6\tau)$ bit operations.
  Furthermore, the polynomials of this RUR have degree at most $d^2$ and bitsize
  in~$\sO(d^2+d\tau)$.
\end{theorem}

%\ifshort
%\begin{definition}\cite{Rou99}
%\label{def:rur}
%Let $I\subset \Q[X,Y]$ be a zero-dimensional ideal and $T=X+aY$ a linear form
%with $a\in \Q$.  The RUR-candidate of $I$ associated to $X+aY$, denoted
%$RUR_{I,a}$, is the following set of four univariate polynomials in~$\Q[T]$
% $$\begin{array}{l}
%\displaystyle f_{I, a} (T) = \prod_{\sigma \in V (I)} (T - X (\sigma) - aY
%    (\sigma))^{\mu_I (\sigma)} \\
%\displaystyle    f_{I, a, v} (T) = \sum_{\sigma \in V (I)} \mu_I (\sigma) v (\sigma)
%    \prod_{\varsigma \in V (I), \varsigma \neq \sigma} (T - X (\varsigma) - aY
%    (\varsigma)),  \quad v\in \{1,X,Y \}
%  \end{array}$$
%  where, for  $\sigma \in V(I)$, $\mu_I(\sigma)$ denotes the multiplicity
%  of $\sigma$ in $I$. 
%\end{definition}
%\fi

\subsection{RUR computation}
\label{sec:RURalgo}
%%%%%%%%%%%%%%%%%%%%%%%%%%%%%%%%%%%%%%%%%%%%%%%%%%%%%%%%%

We show here that the polynomials of a RUR can be expressed as combinations of
specializations of the resultant $R$ and its partial derivatives. The seminal
idea has already been used by several authors in various contexts
(see e.g. \cite{Can,ABRW,Schost})
for computing rational
parameterizations of the radical of a given zero-dimensional ideal and mainly for
bounding the size of a Chow form.
Based on the same idea but keeping track of multiplicities, we present a simple
new 
formulation for the polynomials of a RUR, given separating form.

\begin{proposition}\label{prop:rur-res2}
  For any rational $a$ such that $L_{P}(a)L_{Q}(a)\neq 0$ and such that $X+aY$
  is a separating form of $I=\langle P,Q\rangle$, the RUR of $\langle P,Q\rangle$
  associated to $a$ is as follows:
\[\begin{array}{ll}
\displaystyle f_{I,a}(T)=\frac{R(T,a)}{L_R(a)}
&\quad \displaystyle 
f_{I,a,1}(T)=\frac{f'_{I,a}(T)}{\gcd(f_{I,a}(T),{f'}_{\! I,a}(T))}\\
\displaystyle 
f_{I,a,Y}(T)= \frac{ \frac{\partial R}{\partial S}(T,a) -f_{I,a}(T)\frac{\partial L_R}{ \partial S}(a) }{L_R(a) \gcd(f_{I,a}(T),{f'}_{\! I,a}(T))}
&\quad \displaystyle 
f_{I,a,X}(T)=T f_{I,a,1}(T)-d_T(f_{I,a}) \overline{f_{I,a}(T)}
    -af_{I,a,Y}(T).
\end{array}\]
\end{proposition}
We postpone the proof of Proposition~\ref{prop:rur-res2} to
Section~\ref{sec:proof:lem:rur-res2} and first analyze the complexity 
of the computation of the expressions therein.
Note that a separating form $X+aY$ as in Proposition~\ref{prop:rur-res2} can be computed in
$\sOB(d^8+d^7\tau)$ \cite{bouzidi2013SepElt}.

\begin{proposition}
\label{prop:complexity-rur-cand}
Computing the polynomials in Proposition~\ref{prop:rur-res2} 
can be done with $\sOB(d^7+d^6(\tau+\tau_a))$ bit operations, where $\tau_a$ is the bitsize of $a$.
\end{proposition}

\begin{proof}[Proof of Proposition~\ref{prop:complexity-rur-cand}]
  According to Lemma~\ref{lem:complexity:shear}, the resultant $R(T,S)$ of
  ${P}(T-SY,Y)$ and ${Q}(T-SY,Y)$ with respect to $Y$ has degree $O(d^2)$ in $T$
  and $S$, has bitsize in $\sO(d(d+\tau))$, and that it can be computed in
  $\sO_B(d^6(d+\tau))$ bit operations. We can now apply the formulas of
  Proposition~\ref{prop:rur-res2} for computing the polynomials of the RUR.

  Specializing $R(T,S)$ at $S=a$ can be done by evaluating $O(d^2)$ polynomials
  in $S$, each of degree in $O(d^2)$ and bitsize in $\sO(d^2+d\tau)$. By
  Lemma~\ref{lem:comp:evaluation}, each of the $O(d^2)$ evaluations can be done
  in $\sOB(d^2(d^2+d\tau+\tau_a))$ bit operations and each result has bitsize in
  $\sO(d^2+d\tau+d^2\tau_a)$. Hence, $R(T,a)$ and $f_{I,a}(T)$ have degree in
  $O(d^2)$, bitsize in $\sO(d^2+d\tau+d^2\tau_a)$, and they can be computed with
  $\sOB(d^4(d^2+d\tau+\tau_a))$ bit operations.

  The complexity of computing the numerators of $f_{I,a,1}(T)$ and
  $f_{I,a,Y}(T)$ is clearly dominated by the computation of $\frac{\partial
    R}{\partial S}(T,a)$.  Indeed, computing the derivative $\frac{\partial
    R}{\partial S}(T,S)$ can trivially be done in $O(d^4)$ arithmetic operations
  of complexity $\sOB(d^2+d\tau)$, that is in $\sOB(d^6+d^5\tau)$. Then, as for
  $R(T,a)$, $\frac{\partial R}{\partial S}(T,a)$ has degree in $O(d^2)$, bitsize
  in $\sO(d^2+d\tau+d^2\tau_a)$, and it can be computed within the same
  complexity as the computation of $R(T,a)$.

  On the other hand, since $f_{I,a}(T)$ and ${f'}_{\! I,a}(T)$ have degree in $O(d^2)$ and
    bitsize in $\sO(d^2+d\tau+d^2\tau_a)$, and $f_{I,a}(T)=\frac{R(T,a)}{L_R(a)}$, one can multiply
    these two polynomials by 
 $L_R(a)$ which is of bitsize $\sO(d^2+d\tau+d^2\tau_a)$ and by the denominator of the rational $a$ to the
    power of $d_S(R(T,S))$ which is an integer of bitsize in $O(d^2\tau_a)$, to obtain polynomials
    with coefficients in $\Z$. Hence, according to Lemma~\ref{complexity:gcd}, the gcd of $f_{I,a}(T)$ and ${f'}_{\! I,a}(T)$  can be
    computed in $\sOB(d^4(d^2+d\tau+d^2\tau_a))$ bit operations and it has bitsize in $\sO(d^2+d\tau+d^2\tau_a)$.

    Now, the bit complexity of the division of
  the numerators by the gcd is of the order of the square of their maximum
  degree times their maximum bitsize \cite[Theorem 9.6 and subsequent discussion]{vzGGer}, 
 that is, the divisions (and hence the
  computation of $f_{I,a,1}(T)$ and $f_{I,a,Y}(T)$) can be done in
  $\sOB(d^4(d^2+d\tau+d^2\tau_a))$ bit operations.

  Finally, computing $f_{I,a,X}(T)$ can be done within the same complexity as
  for $f_{I,a,1}(T)$ and $f_{I,a,Y}(T)$ since it is dominated by the computation
  of the squarefree part of $f_{I,a}(T)$, which can be computed similarly and with the same complexity as above,
  by Lemma~\ref{complexity:gcd}.

  The overall complexity is thus that of computing the resultant which is in
  $\sO_B(d^6(d+\tau))$ plus that of computing the above gcd and Euclidean
  division which is in $\sOB(d^4(d^2+d\tau+d^2\tau_a))$. This gives a total of
  $\sOB(d^7+ d^6(\tau+\tau_a))$.
\end{proof}

\subsubsection{Proof of Proposition~\ref{prop:rur-res2}}
\label{sec:proof:lem:rur-res2}

Proposition~\ref{prop:rur-res2} expresses the polynomials $f_{I,a}$ and
$f_{I,a,v}$ of a RUR in terms of specializations (by $S=a$) of the resultant
$R(T,S)$ and its partial derivatives. Since the specializations are done after
considering the derivatives of $R$, we study the relations between these
entities before specializing $S$ by $a$.

For that purpose, we first introduce the following polynomials which are exactly
the polynomials $f_{I,a}$ and $f_{I,a,v}$ of \eqref{eq:defRUR} where the
parameter $a$ is replaced by the variable $S$. These polynomials can be seen as
the RUR polynomials of the ideal $I$ with respect to a ``generic'' linear form
$X+SY$.

\begin{equation}\label{eq:defRUR-gen}
\begin{split}
&\displaystyle{ f_{I} (T,S)} = \prod_{\sigma \in V (I)} (T - X (\sigma) - SY
(\sigma))^{\mu_I (\sigma)} \\
& \displaystyle   { f_{I, v} (T,S)} = \sum_{\sigma \in V (I)} \mu_I (\sigma) v (\sigma)
\prod_{\varsigma \in V (I), \varsigma \neq \sigma} (T - X (\varsigma) - SY (\varsigma)), 
\quad v\in \{1,X,Y \}.
\end{split}
\end{equation}
These polynomials are obviously in $\C[T,S]$, but they are actually in $\Q[T,S]$
because, when $S$ is specialized at any rational value $a$, the specialized
polynomials are those of $RUR_{I,a}$ which are known to be in
$\Q[T]$~(see e.g. \cite{Rou99}).

Before proving Proposition~\ref{prop:rur-res2}, we express the derivatives of
$f_I(T,S)$ in terms of $f_{I,v}(T,S)$, in Lemma~\ref{lem:eq3-4}, and show that
$f_I(T,S)$ is the monic form of the resultant $R(T,S)$, seen as a polynomial in
$T$, in Lemma~\ref{lem:rur-res1}.

\begin{lemma}\label{lem:eq3-4}
Let $g_I(T,S)=\prod_{\sigma \in V (I)} (T - X (\sigma) - SY (\sigma))^{\mu_I
  (\sigma)-1}$.  We have 
 \begin{align}
    \frac{\partial f_I}{\partial T} (T,S) & = g_I(T,S) f_{I,1}(T,S),\label{eq:3}\\
    \frac{\partial f_I}{\partial S} (T,S) &  =g_I(T,S) f_{I,Y}(T,S).\label{eq:4}
 \end{align} 
\end{lemma}
\begin{proof}
  It is straightforward that the derivative of $f_I$ with respect to $T$ is
  $\sum_{\sigma \in V (I)} \mu_I (\sigma) (T - X (\sigma) - SY (\sigma))^{\mu_I
    (\sigma)-1}\prod_{\varsigma \in V (I), \varsigma \neq \sigma}(T - X
  (\varsigma) - SY (\sigma))^{\mu_I (\varsigma)}$, which can be
rewritten as the product of 
$ \prod_{\sigma \in V (I)} (T - X (\sigma) - SY (\sigma))^{\mu_I
  (\sigma)-1}$ and $ \sum_{\sigma \in V (I)} \mu_I (\sigma)
\prod_{\varsigma \in V (I), \varsigma \neq \sigma} (T - X (\varsigma) - SY
(\varsigma)) $
which is exactly the product of $g_I(T,S)$ and $f_{I,1}(T,S)$.

The expression of the derivative of $f_I$ with respect to $S$ is similar to that
with respect to $T$ except that the derivative of $T - X (\sigma) - SY (\sigma)$
is now $Y (\sigma)$ instead of $1$. It follows that $\frac{\partial
  f_I}{\partial S}$
is the product of $ \prod_{\sigma \in V (I)}
(T - X (\sigma) - SY (\sigma))^{\mu_I (\sigma)-1}$ and $ \sum_{\sigma \in V (I)} \mu_I (\sigma) Y(\sigma)
\prod_{\varsigma \in V (I), \varsigma \neq \sigma} (T - X (\varsigma) - SY (\varsigma)) $ which is
 the product of $g_I(T,S)$ and $f_{I,Y}(T,S)$.
\end{proof}

For the proof of Lemma~\ref{lem:rur-res1}, we will need  the following lemma which states that when two polynomials have no common solution at infinity in some
direction, the roots of their resultant with respect to this direction are the
projections of the solutions of the system with cumulated multiplicities.
\begin{lemma}[{\cite[Prop. 2 and 5]{Buse}}]
\label{lem:res-roots}
Let $P, Q\in\F[X,Y]$ defining a zero-dimensional
 ideal $I=\langle P, Q \rangle$, 
such that their leading terms
  $Lc_Y(P)$ and $Lc_Y(Q)$ do not have common roots.  Then
  $Res_Y(P,Q)=c\prod_{\sigma \in V (I)} (X - X (\sigma))^{\mu_I (\sigma)}$ where
  $c$ is nonzero in $\F$.
\end{lemma}

The following lemma links the resultant of $P(T-SY,Y)$ and $Q(T-SY,Y)$ with respect to $Y$ and the polynomial $f_I(T,S)$ as defined above.

\begin{lemma}\label{lem:rur-res1}
$R(T,S)= L_R(S) f_{I}(T,S)$ and, for any $a\in \Q$, $L_{P}(a)L_{Q}(a)\neq 0$ implies that $L_R(a)\neq~0$.
\end{lemma}

\begin{proof}
  The proof is organized as follows. We first prove that for any rational $a$
  such that $L_{P}(a)L_{Q}(a)$ does not vanish, $R(T,a)= c(a) f_I(T,a)$ where
  $c(a)\in\Q$ is a nonzero constant depending on $a$. This is true for
  infinitely many values of $a$ and, since $R(T,S)$ and $f_I(T,S)$ are
  polynomials, we can deduce that $R(T,S)= L_R(S) f_{I}(T,S)$.  This will also
  implies the second statement of the lemma since, if $L_{P}(a)L_{Q}(a)\neq 0$,
  then $R(T,a)= c(a) f_I(T,a)=L_R(a) f_{I}(T,a)$ with $c(a)\neq 0$, thus
  $L_R(a)\neq 0$ (since $f_I(T,a)$ is monic).

  If $a$ is such that $L_{P}(a)L_{Q}(a)\neq 0$, the resultant $R(T,S)$ can be
  specialized at $S=a$, in the sense that $R(T,a)$ is equal to the resultant of
  $P(T-aY,Y)$ and $Q(T-aY,Y)$ with respect to $Y$\cite[Proposition 4.20]{BPR06}.

  We now apply Lemma~\ref{lem:res-roots} to these two polynomials $P(T-aY,Y)$
  and $Q(T-aY,Y)$.  These two polynomials satisfy the hypotheses of this lemma:
  first, their leading coefficients (in $Y$) do not depend on $T$, hence they
  have no common root in $\Q[T]$; second, the polynomials $P(T-aY,Y)$ and
  $Q(T-aY,Y)$ are coprime because $P(X,Y)$ and $Q(X,Y)$ are coprime by
  assumption and the change of variables $(X,Y)\mapsto (T=X+aY,Y)$ is a
  one-to-one mapping (and a common factor will remain a common factor after the
  change of variables). Hence Lemma~\ref{lem:res-roots} yields that $R(T,a)=
  c(a)\, \prod_{\sigma \in V (I_a)} (T - T(\sigma))^{\mu_{I_a} (\sigma)},$ where
  $c(a)\in\Q$ is a nonzero constant depending on $a$, and $I_a$ is the ideal
  generated by $ {P}(T-aY,Y)$ and ${Q}(T-aY,Y)$.

  We now observe that $\prod_{\sigma \in V (I_a)} (T - T(\sigma))^{\mu_{I_a}
    (\sigma)}$ is equal to $f_I(T,a)=\prod_{\sigma \in V (I)} (T -
  X(\sigma)-aY(\sigma))^{\mu_{I} (\sigma)}$ since any solution $(\alpha,\beta)$
  of $P(X,Y)$ is in one-to-one correspondence with the solution
  $(\alpha+a\beta,\beta)$ of $P(T-aY,Y)$ (and similarly for $Q$) and the multiplicities of the solutions also match, i.e. $\mu_{I}
    (\sigma)=\mu_{I_a} (\sigma_a)$ when $\sigma$ and $\sigma_a$ are in
    correspondence through the mapping \cite[\S 3.3 Proposition 3 and
    Theorem~3]{fulton2008algebraic}.   Hence,
  \begin{equation}\label{eq:Rfa}
    L_{P}(a)L_{Q}(a)\neq 0\quad\Rightarrow\quad R(T,a)=c(a) f_I(T,a) \quad \mbox{with}\quad c(a) \neq 0.
  \end{equation}

  Since there is finitely many values of $a$ such that
  $L_{P}(a)L_{Q}(a)L_R(a)=0$ and since $f_I(T,S)$ is monic with respect to $T$,
  \eqref{eq:Rfa} implies that $R(T,S)$ and $f_I(T,S)$ have the same degree in
  $T$, say $D$. We write these two polynomials as
  \begin{equation}\label{eq:RfI}
    R(T,S) = L_R(S)T^D+\sum_{i=0}^{D-1}r_i(S)T^i, \quad\quad\quad
    f_I(T,S)= T^D+\sum_{i=0}^{D-1}f_i(S)T^i.
  \end{equation}
  If $a$ is such that $L_{P}(a)L_{Q}(a)L_R(a)\neq 0$, \eqref{eq:Rfa} and
  \eqref{eq:RfI} imply that $L_R(a)=c(a)$ and $r_i(a)=L_R(a)f_i(a)$, for all
  $i$. These equalities hold for infinitely many values of $a$, and $r_i(S),
  L_R(S)$ and $f_i(S)$ are polynomials in $S$, thus $r_i(S)= L_R(S)f_i(S)$ and,
  by \eqref{eq:RfI}, $R(T,S) =L_R(S)f_I(T,S)$.
\end{proof}

We can now prove Proposition~\ref{prop:rur-res2}, which we recall, for clarity.

\medskip\noindent{\bf Proposition \ref{prop:rur-res2}.}\quad \emph{For any
  rational $a$ such that $ L_{P}(a)L_{Q}(a)\neq 0$ and such that $X+aY$ is a
  separating form of $I=\langle P,Q\rangle$, the RUR of $\langle P,Q\rangle$
  associated to $a$ is as follows:}
\[\begin{array}{ll}
\displaystyle f_{I,a}(T)=\frac{R(T,a)}{L_R(a)}
&\quad \displaystyle 
f_{I,a,1}(T)=\frac{f'_{I,a}(T)}{\gcd(f_{I,a}(T),{f'}_{\! I,a}(T))}\\
\displaystyle 
f_{I,a,Y}(T)= \frac{ \frac{\partial R}{\partial S}(T,a) -f_{I,a}(T)\frac{\partial L_R}{ \partial S}(a) }{L_R(a) \gcd(f_{I,a}(T),{f'}_{\! I,a}(T))}
&\quad \displaystyle 
f_{I,a,X}(T)=T f_{I,a,1}(T)-d_T(f_{I,a}) \overline{f_{I,a}(T)}
    -af_{I,a,Y}(T).
\end{array}\]

\begin{proof}
  Since we assume that $a$ is such that $L_{P}(a)L_{Q}(a)\neq 0$, Lemma~\ref{lem:rur-res1}
  immediately gives the first formula.

  Equation~\ref{eq:3} states that
  $f_{I,1}(T,S)g_I (T,S)= \frac{\partial f_I(T,S)}{\partial T}$,
where $g_I(T,S)=\prod_{\sigma \in V (I)} (T - X (\sigma) - SY (\sigma))^{\mu_I
  (\sigma)-1}$.  
  In addition, $g_I$ being monic in $T$, it never identically vanishes when $S$ is
  specialized, thus the preceding formula yields after  specialization:
  $f_{I,a,1}(T)=\frac{f'_{I,a}(T)}{g_I (T,a)}$.
  Furthermore, $g_I(T,a) = \gcd(f_{I,a}(T), f_{I,a}'(T))$. Indeed, $f_{I,a}(T)= \prod_{\sigma \in V
    (I)} (T - X (\sigma) - aY (\sigma))^{\mu_I (\sigma)}$ and all values $X (\sigma) + aY (\sigma)$,
  for $\sigma\in V (I)$, are pairwise distinct since $X+aY$ is a separating form, thus the gcd of
  $f_{I,a}(T)$ and its derivative is $\prod_{\sigma \in V (I)} (T - X (\sigma) - aY (\sigma))^{\mu_I
    (\sigma)-1}$, that is $g_I(T,a)$. This proves the formula for $f_{I,a,1}$.

  \smallskip Concerning the third equation, Lemma~\ref{lem:rur-res1} together
  with Equation~\ref{eq:4} implies:
  \[\begin{split}
    f_{I,Y}(T,S)&= \frac{\frac{\partial f_I(T,S)}{\partial S}}{g_I (T,S)}=
    \frac{\frac{\partial (R (T,S)/L_R(S))}{\partial S}}{g_I (T,S)}=
    \frac{ \frac{\partial R(T,S)}{ \partial S} L_R(S) -R(T,S) \frac{\partial L_R(S)}{ \partial S}
    }{L_R(S)^2 g_I (T,S)}\\
    &=\frac{ \frac{\partial R(T,S)}{ \partial S}  -f_I(T,S) \frac{\partial L_R(S)}{ \partial S}
    }{L_R(S) g_I (T,S)}.
  \end{split}\] As argued above, when specialized, $g_I (T,a)= \gcd(f_{I,a}(T),
  f_{I,a}'(T))$ and it does not identically vanish. By Lemma~\ref{lem:rur-res1},
  $L_R(a)$ does not vanish either, and the formula for $f_{I,a,Y}$
  follows. \smallskip
  
  It remains to compute $f_{I,a,X}$. 
  % Marc: Commentaire a garder \MP{Note: on aurait pu faire autrement: par
  % unicitÃ© de $f_{I,a,X}$, avec le changement de variables $Y=T-aX$, on aurait
  % inversÃ© les roles de $X$ et $Y$. En pratique ca impose de definir les 2
  % changements de variables... On peut factorizer ce probleme avec le
  % changement de variable $T=aX+bY$ et introduire 2 parametres.}  \FR{Alors on
  % a essyÃ© en fait ... c'est beaucoup plus clair pour les bitsizes de faire
  % avec 2 variables de plus au lieu d'une (voir U-RUR) mais par contre de
  % toutes faÃ§ons on en vient Ã  un truc plus brilolÃ© avec seulement une
  % variable de plus pour l'algorithme sinon on tombe dans la complexitÃ© d'un
  % rÃ©sultant en 4 variables et ce n'est pas bon.}
%
  Definition~\ref{def:rur} implies that, for any root $\gamma$ of~$f_{I,a}$:
  $\gamma=\frac{f_{I,a,X}}{f_{I,a,1}}(\gamma) + a
  \frac{f_{I,a,Y}}{f_{I,a,1}}(\gamma)$, and thus $f_{I,a,X}(\gamma)
  +af_{I,a,Y}(\gamma)-\gamma f_{I,a,1}(\gamma)=0$. Replacing $\gamma$ by $T$, we
  have that the polynomial $f_{I,a,X}(T) +af_{I,a,Y}(T)-T f_{I,a,1}(T)$ vanishes
  at every root of $f_{I,a}$, thus the squarefree part of $f_{I,a}$ divides that
  polynomial. In other words, $f_{I,a,X}(T)=T f_{I,a,1}(T)-af_{I,a,Y}(T) \bmod
  \overline{f_{I,a}(T)}$. We now compute $T f_{I,a,1}(T)$ and $af_{I,a,Y}(T)$
  modulo $\overline{f_{I,a}(T)}$.

  Equation \eqref{eq:defRUR} implies that $f_{I,a,v}(T)$ is equal to $T^{\#
    V(I)-1}\sum_{\sigma \in V (I)} \mu_I (\sigma) v (\sigma) $ plus some terms
  of lower degree in $T$, and that the degree of $\overline{f_{I,a}(T)}$ is $\#
  V(I)$ (since $X+aY$ is a separating form). First, for $v=Y$, this implies that
  $d_T(f_{I,a,Y})< d_T(\overline{f_{I,a}})$, and thus that $af_{I,a,Y}(T)$ is
  already reduced modulo $\overline{f_{I,a}(T)}$.  Second, for $v=1$,
  $\sum_{\sigma \in V (I)} \mu_I (\sigma)$ is nonzero and equal to
  $d_T(f_{I,a})$. Thus, $Tf_{I,a,1}(T)$ and $\overline{f_{I,a}(T)}$ are both of
  degree $\# V(I)$, and their leading coefficients are $d_T(f_{I,a})$ and 1,
  respectively. Hence $T f_{I,a,1}(T)\bmod \overline{f_{I,a}(T)}= T
  f_{I,a,1}(T)-d_T(f_{I,a})\overline{f_{I,a}(T)}$. We thus obtain the last
  equation, that is, $f_{I,a,X}(T)=T
  f_{I,a,1}(T)-d_T(f_{I,a})\overline{f_{I,a}(T)}-af_{I,a,Y}(T)$.
\end{proof}

\subsection{RUR bitsize}\label{sec:size-rur}
%%%%%%%%%%%%%%%%%%%%%%%%%%%%%%%%%%%%%%%%%%%%%%%%%%%%%%%%%

We prove here, in Proposition~\ref{prop:rur-size}, a new bound on the bitsize of the coefficients of the polynomials
of a RUR. This bound is interesting in its own right and is instrumental for
our analysis of the complexity of computing isolating boxes of the solutions of
the input system, as well as for performing \emph{sign\_at} evaluations.
We state our bound for RUR-\emph{candidates}, that is even when the
linear form $X+aY$ is not separating. We only use this result 
when the form is separating, for proving Theorem~\ref{th:rur}, but the general result is interesting in a
probabilistic context when a RUR-candidate is computed with a random linear
form. We also prove our bound, not only for the RUR-candidates of an ideal defined by
\emph{two} polynomials $P$ and $Q$, but for any ideal of  
$\mathbb{Z} [X, Y]$ that contains $P$ and $Q$ (for instance the radical of
$\ideal{P,Q}$  or the ideals  obtained by decomposing $\ideal{P,Q}$ 
according to the multiplicity of the solutions).

\begin{proposition}
  \label{prop:rur-size} 
  Let $P, Q \in \mathbb{Z} [X, Y]$ be two coprime polynomials of total degree at
  most $d$ and maximum bitsize $\tau$, let $a$ be a rational of bitsize
  $\tau_a$, and let $J$ be any ideal of $\mathbb{Z} [X, Y]$  containing $P$ and $Q$. The polynomials of the RUR-candidate of $J$ associated to
  $a$ have degree at most $d^2$ and bitsize in $\sO (d^2\tau_a +d\tau )$. Moreover, there exists
    an integer of bitsize in $\sO(d^2\tau_a+d\tau)$ such that the product of
    this integer with any polynomial in the RUR-candidate yields a polynomial
    with integer coefficients.\footnote{In other words,  the mapping 
$\gamma \mapsto
\left(\frac{f_{J,a,X}}{f_{J,a,1}}(\gamma),\frac{f_{J,a,Y}}{f_{J,a,1}}(\gamma)\right)$ sending the solutions of $f_{J,a}(T)$ to those of $J$ (see Definition~\ref{def:rur}) can be
    defined with polynomials with integer coefficients of bitsize  $\sO(d^2\tau_a+d\tau)$. This will
    be needed in the proof of Lemma \ref{lem:f_F}.}
\end{proposition}

Before proving Proposition \ref{prop:rur-size}, we prove a corollary of
Mignotte's lemma stating that the bitsize of a factor of a polynomial $P$ with
integer coefficients does not differ much than that of $P$.  We also recall a
notion of primitive part for polynomials in $\Q[X,Y]$ and some of its
properties.

\begin{lemma}[Mignotte]\label{lem:mignotte}
  Let $P\in \Z[X,Y]$ be of degree at most $d$ in each variable with coefficients
  bitsize at most $\tau$. If $P=Q_1Q_2$ with $Q_1$, $Q_2$ in $\Z[X,Y]$, then the
  bitsize of $Q_i$, $i=1,2$, is in~$\sO(d+\tau)$.
\end{lemma}
\begin{proof}
  A polynomial can be seen as the vector of its coefficients and we denote by
  $\|P\|_k$ the $L^k$ norm of $P$.  Mignotte lemma \cite[Theorem 4bis
  p. 172]{Mignotte} states 
  that $\|Q_1\|_1 \|Q_2\|_1 \leq 2^{2d} \|P\|_2$. One always has
  $\|Q_i\|_\infty\leq \|Q_i\|_1$ and since the polynomials have integer
  coefficients, $1\leq \|Q_i\|_\infty$. Thus $\|Q_j\|_\infty \leq 2^{2d}
  \|P\|_2$ and $\log\|Q_j\|_\infty \leq {2d}+\log\|P\|_2$.  Thus, by definition,
  the bitsize of $Q_j$ is $\lfloor \log\|Q_j\|_\infty\rfloor+1\leq
  {2d}+1+\log\|P\|_2$. Since $P$ has degree at most $d$ in each variable, it has
  at most $(d+1)^2$ coefficients which are bounded by $2^\tau$, thus $\|P\|_2 <
  \sqrt{(d+1)^2 2^{2\tau}}$ which yields that the bitsize of $Q_j$ is less than
  $2d+1+\log(d+1) +{\tau}$.
\end{proof}

\medskip\noindent\emph{Primitive part.}\quad Consider a polynomial $P$ in
$\Q[X,Y]$ of degree at most $d$ in each variable.  It can be written
$P=\sum_{i,j=0}^{d}\frac{a_{ij}}{b_{ij}}X^iY^j$ with $a_{ij}$ and $b_{ij}$
coprime in $\Z$ for all $i,j$. We define the \emph{primitive part} of $P$,
denoted $pp(P)$, as $P$ divided by the gcd of the $a_{ij}$ and multiplied by the
least common multiple (lcm) of the $b_{ij}$. (Note that this definition is not
entirely standard since we do not consider contents that are polynomials in $X$
or in $Y$.)  We also denote by $\tau_P$ the bitsize of $P$ (that is, the maximum
bitsize of all the $a_{ij}$ and $b_{ij}$).  We prove three properties of the
primitive part which will be useful in the proof.

\begin{lemma}\label{lem:primpart}
For any two polynomials $P$ and $Q$ in $\Q[X,Y]$, we have the following properties: (i)
$pp(PQ)=pp(P)\, pp(Q)$. (ii) If $P$ is monic then $\tau_{P}\leq \tau_{pp(P)}$ and, more generally, 
if $P$ has one  coefficient, $\xi$,  of bitsize $\tau_\xi$, then
$\tau_{P}\leq \tau_\xi+\tau_{pp(P)}$. (iii)  If $P$ has coefficients in $\Z$, then
$\tau_{pp(P)}\leq \tau_{P}$.
\end{lemma}
\begin{proof}
Gauss Lemma states that if two univariate polynomials with integer coefficients are
primitive, so is their product. This lemma can straightforwardly be extended to be used in our
context by applying a change of variables of the form $X^iY^j\rightarrow Z^{ik+j}$ with
$k>2\max(d_Y(P),d_Y(Q))$. 
Thus, if $P$ and $Q$ in $\Q[X,Y]$ are primitive (i.e., each of
them has integer coefficients whose common gcd is 1), their product is primitive.
It follows that $pp(PQ)=pp(P)\,pp(Q)$ because, writing
  $P=\alpha\,pp(P)$ and $Q=\beta\,pp(Q)$, we have
  $pp(PQ)=pp(\alpha\,pp(P)\,\beta\,pp(Q)) = pp(pp(P)\,pp(Q))$ which is equal to
  $pp(P)\,pp(Q)$ since 
the product of two primitive polynomials is primitive.

  Second, if $P\in\Q[X,Y]$ has one coefficient, $\xi$, of bitsize $\tau_\xi$,
  then $\tau_{P}\leq \tau_\xi+\tau_{pp(P)}$.  Indeed, We have
  $P=\xi\frac{P}{\xi}$ thus $\tau_{P}\leq \tau_\xi+\tau_{\frac{P}{\xi}}$. Since
  $\frac{P}{\xi}$ has one of its coefficients equal to 1, its primitive part is
  $\frac{P}{\xi}$ multiplied by an integer (the lcm of the denominators), thus
  $\tau_{\frac{P}{\xi}}\leq \tau_{pp(\frac{P}{\xi})}$ and
  $pp(\frac{P}{\xi})=pp(P)$ by definition, which implies the claim.

  Third, if $P$ has coefficients in $\Z$, then $\tau_{pp(P)}\leq \tau_{P}$ since
  $pp(P)$ is equal to $P$ divided by an integer (the gcd of the integer
  coefficients).
\end{proof}

  The idea of the proof of Proposition~\ref{prop:rur-size}  is,
  for $J\supseteq I=\ideal{P,Q}$, to first argue that polynomial $f_J$, that is the first
  polynomial of the RUR-candidate before specialization at $S=a$, is a factor of $f_I$ which is a
  factor of   the resultant $R(T,S)$ by Lemma \ref{lem:rur-res1}. 
We then derive a bound of $\sO(d^2+d\tau)$ on the bitsize of $f_{J}$ 
  from the bitsize of this resultant using Lemma~\ref{lem:mignotte}.
The bound on the bitsize of the other polynomials of the  non-specialized RUR-candidate of $J$ follows from
  the bound on $f_J$ and we  finally specialize all these polynomials at $S=a$ which
  yields the result. We decompose this proof in two lemmas to emphasize that, although the bound on the bitsize of $f_{J}$  uses
the fact that $J$ contains polynomials $P$ and $Q$, the second part  of the proof only uses the
bound on $f_{J}$.

\begin{lemma}\label{lem:rur-size1}
  Let $P, Q \in \mathbb{Z} [X, Y]$  be two coprime polynomials of total degree at
  most $d$ and maximum bitsize $\tau$, and  $J$ be any ideal of $\mathbb{Z} [X, Y]$  containing $P$ and $Q$. 
Polynomials $f_{J}(T,S)$ (see \eqref{eq:defRUR-gen})  and its primitive part have bitsize in $\sO
(d^2 +d\tau )$ and degree at most $d^2$ in each variable.
\end{lemma}
\begin{proof}
Consider  an ideal $J$  containing $I=\ideal{P,Q}$. 
Counted with multiplicity, the set of solutions of
$J$ is a subset of those of $I$ thus, by Equation \eqref{eq:defRUR-gen}, polynomial
$f_J(T,S)$ is monic in $T$ and $f_J(T,S)$ divides $f_I(T,S)$. Furthermore, $f_I(T,S)$ divides
$R(T,S)$ by Lemma~\ref{lem:rur-res1}. Thus $f_J(T,S)$ divides $R(T,S)$ and we consider  $h\in\Q[T,S]$  such that $f_J\,
h=R$. Taking the primitive part, we have $pp(f_J)\, pp(h)=pp(R)$ by Lemma~\ref{lem:primpart}. The
bitsize of $pp(R)$ is in $\sO(d^2+d\tau)$ because $R$ is of bitsize $\sO(d^2+d\tau)$ (Lemma~\ref{lem:complexity:shear}) and, since $R$ has integer coefficients, $\tau_{pp(R)}\leq \tau_{R}$
(Lemma~\ref{lem:primpart}). This implies that $pp(f_J)$  also has bitsize in $\sO(d^2+d\tau)$ by
Lemma~\ref{lem:mignotte} because the degree of $pp(R)$ is in $O(d^2)$
(Lemma~\ref{lem:complexity:shear}).
Furthermore, since  $f_J(T,S)$ is monic in $T$, $\tau_{f_J}\leq \tau_{pp(f_J)}$
(Lemma~\ref{lem:primpart}) which implies that both $f_J$ and its primitive part have bitsize
in~$\sO(d^2+d\tau)$. 
Finally, the number of solutions (counted with multiplicity) of $\ideal{P,Q}$ is at most $d^2$ by
the Bézout bound, and this bound also holds for $J\supseteq \ideal{P,Q}$. It then follows from Equation
\eqref{eq:defRUR-gen} that $f_J$ has degree at most $d^2$ in each variable. 
\end{proof}

\begin{lemma}\label{lem:rur-size2}
Let $J$ be any ideal such that  polynomials $f_{J}(T,S)$ (see \eqref{eq:defRUR-gen}) and its primitive part have degree $O(d^2)$
and bitsize in $\sO (d^2 +d\tau )$ and $a$ is a rational of bitsize $\tau_a$. Then all the polynomials of the RUR-candidate $RUR_{J,a}$ have bitsize in $\sO (d^2\tau_a +d\tau )$.
 Moreover, there exists
    an integer of bitsize in $\sO(d^2\tau_a+d\tau)$ such that its  product with any polynomial in the RUR-candidate yields a polynomial
    with integer coefficients.
\end{lemma}
\begin{proof}
%  \medskip\noindent
\emph{Bitsize of $f_{J,v}$, $v\in\{1,Y\}$.}\quad We consider
  the equations of Lemma~\ref{lem:eq3-4} which can be written as $\frac{\partial
    f_J}{\partial u} (T,S) = g_J(T,S) f_{J,v}(T,S)$ where $u$ is $T$ or $S$, and
  $v$ is $1$ or $Y$, respectively.  We first bound the bitsize of one
  coefficient, $\xi$, of $f_{J,v}$ so that we can apply
  Lemma~\ref{lem:primpart} which states that $\tau_{f_{J,v}}\leq
  \tau_\xi+\tau_{pp(f_{J,v})}$. We consider the leading coefficient $\xi$ of
  $f_{J,v}$ with respect to the lexicographic order $(T,S)$. Since $g_J$ is
  monic in $T$ (see Lemma~\ref{lem:eq3-4}), the leading coefficient (with
  respect to the same ordering) of the product $g_Jf_{J,v}=\frac{\partial
    f_J}{\partial u}$ is $\xi$ which thus has bitsize in $\sO(\tau_{f_J})$ 
  (since it is bounded by $\tau_{f_J}$ plus the log of the degree of
  $f_J$). It thus follows from the hypothesis on $\tau_{f_J}$  that  $\tau_{f_{J,v}}$ is in $\sO( d^2+d\tau+\tau_{pp(f_{J,v})})$.

  We now take the primitive part of the above equation (of Lemma~\ref{lem:eq3-4}), which gives
  $pp(\frac{\partial f_J}{\partial u} (T,S)) = pp(g_J(T,S))\ pp(f_{J,v}(T,S))$. By
  Lemma~\ref{lem:mignotte}, $\tau_{pp(f_{J,v})}$ is in $\sO(d^2+\tau_{pp(\frac{\partial
      f_J}{\partial u})})$. In order to bound the bitsize of $pp(\frac{\partial f_J}{\partial u})$,
  we multiply $\frac{\partial f_J}{\partial u}$ by the lcm of the denominators of the coefficients of
  $f_J$, which we denote by $\lcm_{f_J}$.
Multiplying by a constant does not change the primitive part and $\lcm_{f_J} \frac{\partial
f_J}{\partial u}$ has integer coefficients, so 
the bitsize of $pp(\frac{\partial
    f_J}{\partial u})=pp(\lcm_{f_J}\,\frac{\partial f_J}{\partial u})$ is thus at most that of
  $\lcm_{f_J}\,\frac{\partial f_J}{\partial u}$ which is bounded by the sum of the bitsizes of
  $\lcm_{f_J}$ and $\frac{\partial f_J}{\partial u}$. 
By hypothesis, the bitsize of $f_J$ is in $\sO(d^2+d\tau)$ so the bitsize of $\frac{\partial
f_J}{\partial u}$ is also in $\sO(d^2+d\tau)$. On the other hand, since $f_J$ 
is monic (in $T$), $f_J\,\lcm_{f_J}=pp(f_J)$ and $\tau_{\lcm_{f_J}}\leq
  \tau_{pp(f_J)}$ which is in $\sO(d^2+d\tau)$ by hypothesis. 
  It follows that $\tau_{pp(f_{J,v})}$ and $\tau_{f_{J,v}}$ are also in
  $\sO(d^2+d\tau)$ for $v\in\{1,Y\}$.

  \medskip\noindent\emph{Bitsize of $f_{J,X}$.}\quad We obtain the bound for
  $f_{J,X}$ by symmetry. Similarly as we proved that $f_{J,Y}$ has bitsize in
  $\sO(d^2+d\tau)$, we get, by exchanging the role of $X$ and $Y$ in
  Equation~\eqref{eq:defRUR-gen} and Lemma~\ref{lem:eq3-4}, that $ \sum_{\sigma
    \in V (J)} \mu_J (\sigma) X (\sigma) \prod_{\varsigma \in V (J), \varsigma
    \neq \sigma} (T - Y (\varsigma) - SX (\varsigma))$ has bitsize in
  $\sO(d^2+d\tau)$. This polynomial is of degree $O(d^2)$ in $T$ and $S$, by hypothesis, thus
  after replacing $S$ by $\frac{1}{S}$ and then $T$ by $\frac{T}{S}$, the
  polynomial is of degree $O(d^2)$ in $T$ and $\frac{1}{S}$. We multiply it by
  $S$ to the power of $\frac{1}{S}$ and obtain $f_{J,X}$ which is thus of
  bitsize $\sO(d^2+d\tau)$.

  \medskip\noindent\emph{Specialization at $S=a$.}\quad To bound the bitsize of
  the polynomials of $RUR_{J,a}$ (Definition~\ref{def:rur}), it remains to
  evaluate the polynomials $f_{J}$ and $f_{J,v}$, $v\in\{1,X,Y\}$, at the
  rational value $S=a$ of bitsize $\tau_a$.  Since these polynomials have degree
  in $S$ in $O(d^2)$ and bitsize in $\sO(d^2+d\tau)$, it is straightforward that
  their specializations at $S=a$ have bitsize in
  $\sO(d^2+d\tau+d^2\tau_a)=\sO(d^2\tau_a+d\tau)$.

  \medskip\noindent\emph{The $\lcm$ of the denominators of all the coefficients in the polynomials
  of $RUR_{J,a}$ has       bitsize $\sO(d^2\tau_a+d\tau)$.}\quad We have already argued that
  $\lcm_{f_J}$, the lcm
  of the denominators of the coefficients of $f_{J}$, is in $\sO(d^2+d\tau)$. For each of the other
  polynomials $f_{J,v}$, $v\in\{1,X,Y\}$, 
 denote by $\lcm_{f_{J,v}}$ and $\gcd_{f_{J,v}}$ the lcm of
  the denominators of its coefficients and the gcd of their numerators. By definition,
  $pp(f_{J,v})=\frac{\lcm_{f_{J}}}{\gcd_{f_{J,v}}} f_{J,v}$. Let $c$ be any coefficient of
  $pp(f_{J,v})\in\Z[S,T]$  and $\frac{a}{b}$ be the corresponding coefficient of
  $f_{J,v}\in\Q[S,T]$ (with $a$ and $b$ coprime integers); we have
$\lcm_{f_{J}} =  c\,\frac{b}{a}\gcd_{f_{J,v}}\leq c\,b$ since $\gcd_{f_{J,v}}$ divides $a$. It
follows that $\tau_{\lcm_{f_{J}}}\leq \tau_{pp(f_{J,v})}+\tau_{f_{J,v}}$ which are both in
$\sO(d^2+d\tau)$, as proved above.  Hence the $\lcm$ of the denominators of all the
coefficients in $RUR_{J,a}$ has       bitsize $\sO(d^2+d\tau)$. Finally, since all these polynomials have degree
$O(d^2)$, when specializing by $S=a$, the bitsize of the denominators of the coefficients of the
polynomials increase by at most $O(d^2\tau_a)$ and thus the bitsize of their lcm also increases by
at most $O(d^2\tau_a)$, which concludes the proof. 
\end{proof}

\begin{proof}[Proof of Proposition~\ref{prop:rur-size}]
By Lemma~\ref{lem:rur-size1}, $f_J$ has degree at most $d^2$ in each variable, so has $f_{J,v}$,
$v\in\{1,X,Y\}$ by Equation
\eqref{eq:defRUR-gen}. It follows from Equation~\eqref{eq:defRUR} that  all the polynomials of any
RUR-candidate of $J$  have degree at most $d^2$.
The rest of the proposition
is a corollary of Lemmas~\ref{lem:rur-size1} and~\ref{lem:rur-size2}.
\end{proof}

\section{Applications}\label{sec:applications}
%%%%%%%%%%%%%%%%%%%%%%%%%%%%%%%%%%%%%%%%%%%%%%%%%

We present three applications enlightening the advantages of computing a RUR of
a system. 
The first one is the isolation of the solutions, that is computing boxes with
rational coordinates that isolate the solutions.  The second one is the
evaluation of the sign of a bivariate polynomial at a real solution of the
system. Finally, we address the problem of computing a rational parameterization
of a system defined by several equality and inequality constraints.
In all these applications, we take advantage of the RUR to transform bivariate
operations on the system into univariate operations.  We assume
that the polynomials of the RURs satisfy the bitsize bound of
Theorem~\ref{th:rur}.

%\subsection{Preliminaries}
%%%%%%%%%%%%%%%%%%%%%%%%%%%%%%%%%%%%%%%%%%%%%%%%%%
 We start by recalling the complexity of isolating the real roots of a univariate
polynomial. Here, $f$ denotes a univariate polynomial of degree $d$
with integer coefficients of bitsize at most~$\tau$.

\begin{lemma}[{\cite[Theorem 10]{sagraloff2012issacNewDsc}\footnote{Theorem 10
      of \cite{sagraloff2012issacNewDsc} states that isolating the real roots of
      $f$ and refining \emph{all} the isolating intervals up to a precision of
      $L$ bits can be done with $\sO_B(d^3\tau+d^2L)$ bit operations. However,
      its proof establishes the stronger result, which we stated in
      Lemma~\ref{comp:isolation}.
      % result is actually proved for the refinement of any single isolating
      % interval.
      Note that the proof is currently only available in the manuscript
      corresponding to \cite{sagraloff2012issacNewDsc} which is available on the
      author's webpage.}}]
  \label{comp:isolation}
  Let $f$ be squarefree. The bit complexity of isolating all the real roots of
  $f$ is in $\sO_B(d^3\tau)$.  Then, the bit complexity of refining any of these
  isolating intervals up to a precision of $L$ bits is in~$\sO_B(d^2\tau+dL)$.
\end{lemma}

\begin{lemma}[{\cite[Theorem 4]{1979rump}}]
%\cite[Corollary 6.29, est-ce que c'est ok pour le cas non squarefree?]{Yap-2000}
  \label{lem:sepbound}
  Let the minimum root separation bound of $f$ (or simply the separation bound
  of $f$) be the minimum distance between two different complex roots of $f$:
  $\text{sep}(f)=\min_{ \{\gamma,\, \delta \text{ roots of }f,\ \gamma\neq
    \delta \}}|\gamma-\delta|$. One has $sep(f)> 1/( 2d^{d/2+2}(d2^\tau+1)^d )$,
  which yields $sep(f) > 2^{-\sO(d\tau)}$.
\end{lemma}

\subsection{Computation of isolating boxes}\label{sec:boxes}

By Definition \ref{def:rur}, the RUR of an ideal $I$ defines a mapping
between the roots of a univariate polynomial and the solutions of $I$, {which} %this
yields an algorithm to compute isolating boxes.
Given a RUR of the ideal $I$, $\{ f_{I, a}, f_{I, a, 1}, f_{I, a, X}, f_{I, a,
  Y} \}$, isolating boxes for the real solutions can be computed by first
computing isolating intervals for the real roots of the univariate polynomial
$f_{I, a}$ and then, evaluating the rational fractions $\frac{f_{I, a, X}}{f_{I,
    a, 1}}$ and $\frac{f_{I, a,Y}}{f_{I, a, 1}}$ by interval arithmetic.
However, for the simplicity of the proof, instead of evaluating by interval {arithmetic} each
of these fractions of polynomials, we instead compute the product of its
numerator with the inverted denominator modulo $f_{I, a}$, and then evaluate
this resulting polynomial on the isolating intervals of the real roots of $f_{I,
  a}$ (note that we obtain the same complexity bound if we directly evaluate the
fractions, but the proof is more 
 technical, although not difficult, and we
omit it here).  When these isolating intervals
are sufficiently refined, the computed boxes are necessarily disjoint and thus
isolating. The following proposition analyzes the bit complexity of this
algorithm.

\begin{proposition}\label{prop:computing-boxes}
  Given a RUR of $\ideal{P,Q}$, isolating boxes for the solutions of
  $\ideal{P,Q}$ can be computed in $\sOB(d^8+d^7\tau)$ bit operations, where $d$
  bounds the total degree of $P$ and $Q$,
  and $\tau$ bounds the bitsize of their coefficients.  The vertices of these
  boxes have bitsize in $\sO(d^3\tau)$.
\end{proposition}

\begin{proof}
  For every real solution $\alpha$ of $I=\langle P, Q \rangle$, let $J_{X,
    \alpha} \times J_{Y, \alpha}$ be a box containing it. A sufficient
  condition for these boxes to be isolating is that the width of every interval
  $J_{X, \alpha}$ and $J_{Y, \alpha}$ is less than half the separation bound of
  the resultant of $P$ and $Q$ with respect to $X$ and $Y$, respectively.
  Such a resultant has degree at most $2d^2$ and bitsize in $\sO(d\tau)$ 
  by \cite[Proposition 8.46]{BPR06}.  Lemma~\ref{lem:sepbound} thus yields
  a lower bound of $2^{-\varepsilon}$ with $\varepsilon$ in $\sO(d^3\tau )$ on
  the separating bound of such a resultant.  It is thus sufficient
  to compute, for every $\alpha$, a box $J_{X, \alpha} \times J_{Y, \alpha}$
  that contains $\alpha$ and such that the widths of these intervals are
   smaller than half of $2^{-\varepsilon}$.
 For clarity and technical reasons, we define
     $\varepsilon'=\varepsilon+2$.
   In fact, an explicit value of $\varepsilon$ is not needed to compute
    isolating boxes since the algorithm uses adaptive refinements of the boxes
    and a test of box disjointness. On the other hand, an explicit value of
    $\varepsilon$ will be used to reduce the bitsize of the box endpoints and an
    asymptotic estimate will be used for the complexity analysis. 
  More precisely, the algorithm proceeds as follows. First, the real roots of
  ${f_{I, a}}$ are isolated. Then, we refine these intervals and, during the
  refinement, we routinely evaluate the polynomials of the mapping at these
  intervals, and we stop when all the resulting boxes
  are pairwise disjoint. It is of course critical not to evaluate the
  polynomials of the mapping to often; for every real root of ${f_{I, a}}$, we
  perform these evaluations every time the number of identical consecutive {first} bits
  of the two interval boundaries doubles or, in other words, every time the
  width of the interval becomes smaller than $2^{-2^k}$ for some positive
  integer $k$.

  According to Definition \ref{def:rur}, given a RUR $\{ f_{I, a}, f_{I, a, 1},
  f_{I, a, X}, f_{I, a, Y} \}$ of $I$, the mapping $\gamma\mapsto
  \left(\frac{f_{I,a,X}}{f_{I,a,1}}(\gamma),\frac{f_{I,a,Y}}{f_{I,a,1}}(\gamma)\right)$
  defines a one-to-one correspondence between the real roots of $f_{I, a}$ and
  those of~$I$. Thus every isolating interval $J_\gamma$ of the real roots of
  $f_{I, a}$ is mapped through this mapping to a box that contains the
  corresponding solution of~$I$.  We first show how to modify this rational
  mapping into a polynomial one.
  Second, we bound, in terms of the width of $J_\gamma$, the side length of the
  box obtained by interval arithmetic as the image of $J_\gamma$ through the
  mapping. We will then deduce an upper bound on the width of $J_\gamma$ that
  ensures that the side length of its box image is less than
  $2^{-\varepsilon'}$. 
This thus  gives a worst-case refinement precision on the
  isolating intervals of $f_{I, a}$ for the boxes to be disjoint. We then
  analyze the complexity of the proposed algorithm.

  \medskip\noindent\emph{Polynomial mapping.}\quad By
  Proposition~\ref{prop:rur-res2}, the polynomials $f_{I, a}$ and $f_{I, a, 1}$
  are coprime and thus $f_{I, a, 1}$ is invertible modulo $f_{I, a}$.  The
  rational mapping can thus be transformed into a polynomial one by replacing
  $\frac{1}{f_{I,a,1}}$ by $\frac{1}{f_{I,a,1}}\mod f_{I, a}$. 
Since
  $\frac{1}{f_{I,a,1}}$ and its inverse modulo $f_{I, a}$ coincide when $f_{I,
    a}$ vanishes (by Bézout's identity), this polynomial mapping still maps the
  real roots of $f_{I, a}$ to those of $I$.

  This polynomial mapping can be computed in $\sOB(d^6+d^5\tau)$ bit operations and these
  polynomials have degree
  less than $4d^2$ and bitsize in $\sO(d^4+d^3\tau)$. Indeed, the bit complexity of computing the
  inverse $\frac{1}{f_{I,a,1}}$ modulo $f_{I, a}$ is soft linear in the square of their maximum
  degree times their maximum bitsize \cite[Corollary
  11.11(ii)]{vzGGer},\footnote{{\cite[Corollary 11.11(ii)]{vzGGer} applies because this
      inverse is the cofactor of $f_{I,a,1}$ in the last line of the extended Euclidean algorithm
      corresponding to the resultant of $f_{I,a,1}$ and $f_{I,a}$. Note that this assumes that
      $f_{I,a,1}$ and $f_{I,a}$ have integer coefficients but this is not an issue because, by
      Proposition~\ref{prop:rur-size}, all polynomials of the RUR can be transformed into integer
      polynomials with the same asymptotic bitsize by multiplying them by one and the same
      integer.}}
%\MP{je passe
%    sous silence le fait qu'il faut revenir dans Z pour applique le corollary et
%    que ca change pas la bitsize grace à notre Proposition~\ref{prop:rur-size},
%    est-ce qu'on veut aussi ces details?}
  which yields a complexity of $\sOB((d^2)^2 (d^2+d\tau))$ by
  Theorem~\ref{th:rur}.  The bitsize of this inverse is soft linear in the
  product of their maximum degree and maximum bitsize \cite[Corollary
  6.52]{vzGGer}, that is $\sO(d^2 (d^2+d\tau))$.  Furthermore, the product of
  this inverse and of $f_{I, a, X}$ or $f_{I, a, Y}$ can also be done with a bit
  complexity that is soft linear in the product of their maximum degree and
  maximum bitsize \cite[Corollary 8.27]{vzGGer}, that is in
  $\sOB(d^2(d^4+d^3\tau))$. This concludes the proof of the claim since the
  degree of the inverse modulo $f_{I, a}$ is less than that of $f_{I, a}$ and
  all the polynomials of the RUR have degrees at most $d^2$
  by Theorem~\ref{th:rur}.
  %% NOTE: the time complexity and bit complexity of the EEA do refer to the same algo (cf p.178)

  \medskip\noindent\emph{Width expansion through interval arithmetic
    evaluation.}\quad
  We recall a standard straightforward property of interval arithmetic for
  polynomial evaluation.  We consider here exact interval arithmetic, that is,
  the arithmetic operations on the interval boundaries are considered exact.
  Let $J = [ a, b] $ be an interval with rational endpoints such that $\max ( |
  a |, | b | ) \leqslant 2^{\sigma}$ and let $f \in \mathbbm{Z} [ T ]$ be a
  polynomial of degree $d_f$ with coefficients of bitsize $\tau_f$. Denoting the
  width of $J$ by $w ( J ) = | b - a |$, $f(J)$ can be evaluated by interval
  arithmetic into an interval $f_\square(J)$ whose width is at most $2^{\tau_f +
    d_f \sigma} d_f^2 w ( J )$ (see e.g. \cite[Lemma
  8]{isotop-mcs-10}).\footnote{For completeness, we recall the proof which is
    rather straightforward.  We apply basic formulas for the sum and the product
    of intervals \cite[Theorem 9, p.15]{ga-jh-iic-83}. For any real number $a$
    and integer $n\geq 1$,\ $w(A\pm B)=w(A)+w(B),$\ \ $w(aA)=|a|w(A),$\ \
    $w(AB)\leq w(A)|B| + |A|w(B), $ and $w(A^n)\leq n|A|^{n-1}w(A)$.
%\[  \begin{array}{ll}
%    w(A\pm B)=w(A)+w(B), &
%\quad\quad w(aA)=|a|w(A),\\
%w(AB)\leq w(A)|B| + |A|w(B), &
%\quad\quad w(A^n)\leq n|A|^{n-1}w(A).
%  \end{array}\]
  Writing   $f(T)=\sum_{i=0}^{d_f} c_iT^i$ with $|c_i|\leq 2^{\tau_f}$, we have
  \begin{eqnarray*}
    w(f_\square(J)) & =& \sum_{i=1}^{d_f} |c_i|w(J^i)
    \quad \leq\quad  2^{\tau_f} \sum_{i=1}^{d_f} i|J|^{i-1}w(J)
    \quad \leq\quad   2^{\tau_f} w(J)d_f \sum_{i=1}^{d_f} |J|^{i-1}\\
    & \leq & 2^{\tau_f} w(J)d_f^2 \max(1,|J|^{d_f-1})
    \quad \leq\quad    2^{\tau_f} w(J)d_f^2 2^{d_f\sigma}.
  \end{eqnarray*}
 }
In other words, if $w ( J ) \leq 2^{-  \varepsilon'   -
    \tau_f - d_f \sigma- 2 \log  d_f }$, then $w ( f_\square( J ) ) \leq 2^{-  \varepsilon' }$.

  We now apply this property on the polynomials of the mapping evaluated on
  isolating intervals of $f_{I, a}$. We denote by $d_f$ and $\tau_f$ the maximum
  degree and bitsize of the polynomials of the mapping; as shown
  above
  $d_f<4d^2$ and $\tau_f\in\sO(d^4+d^3\tau)$.  The polynomial $f_{I, a}$ has
  bitsize $\tau_{f_{I,a}}$ in $\sO(d^2+d\tau)$ (Theorem~\ref{th:rur}), thus, by
  Cauchy's bound (see e.g.  \cite[\S 6.2]{Yap-2000}), the maximum absolute value
  of its roots is smaller than $1+ 2^{2\tau_{f_{I,a}}}$.  Considering intervals
  of isolation for $f_{I, a}$ whose widths are bounded by a constant, we thus
  have that the maximum absolute value of the boundaries of the isolating
  intervals are smaller than $2^\sigma$ with $\sigma=\sO(d^2+d\tau)$.
  Now, consider any isolating interval of $f_{I, a}$ of width less than $2^{-
    \varepsilon' - \tau_f - d_f \sigma- 2 \log d_f }$.  The above property
  implies that we can evaluate by interval arithmetic the polynomials of the
  mapping on any such intervals and obtain an interval of width less than
  $2^{-\varepsilon'}$.  In other words, the worst-case refinement precision of
  the isolating intervals of $f_{I, a}$ for the boxes to be disjoint is
  $L=\varepsilon' + \tau_f + d_f \sigma+ 2 \log d_f$. In addition, since
  $\varepsilon'$ is in $\sO(d^3\tau)$, $L$ is in $\sO(d^4+d^3\tau)$.

 \medskip\noindent\emph{Analysis of the algorithm.}\quad For isolation
  and refinement, we consider the polynomial $\overline{pp(f_{I, a})}$, instead
  of $f_{I, a}$, which is also of degree bounded by $d^2$ and bitsize in
  $\sO(d^2+d\tau)$. Indeed, Proposition~\ref{prop:rur-size} implies that the
  integer polynomial $pp(f_{I, a})$ has bitsize in $\sO(d^2+d\tau)$ and
  Lemma~\ref{complexity:gcd} yields that its squarefree part (which the gcd-free
  part of itself and its derivative) is of the same bitsize and can be computed
  in $\sO(d^6+d^5\tau)$.  According to Lemma~\ref{comp:isolation}, the first
  step of the algorithm, the isolation of the roots of $\overline{pp(f_{I, a})}$
  can be done in $\sOB(d^8+d^7\tau)$ bit operations.
  Then, according to the above discussion, these roots will be refined to a
  maximum precision  $L=\sO(d^4+d^3\tau)$. Again, Lemma~\ref{comp:isolation}
  yields a complexity of $\sOB((d^2)^3(d^2+d\tau) + (d^2)^2L)=\sOB(d^8+d^7\tau)$
  for all these refinements.

  It remains to analyze the cost of the evaluations of the mapping and the cost
    of the box-disjointness tests.
For a given root, an evaluation of the polynomials of the
  mapping is performed each time its isolating interval precision is doubled,
  the number of evaluations is thus logarithmic in the maximum precision reached, {that is}
  $L$. One evaluation by interval arithmetic of the polynomials of the mapping,
  which have degree $O(d^2)$ and bitsize $\sO(d^4+d^3\tau)$, on one isolating
  intervals whose endpoints have bitsize at most $L\in \sO(d^4+d^3\tau)$ can be
  done in $\sOB(d^2(d^4+d^3\tau))$ bit operations by
  Lemma~\ref{lem:comp:evaluation} and the resulting intervals have endpoints of
  bitsize in $\sO(d^2(d^4+d^3\tau))$. The cost of the $O(\log L)$ evaluations
  for the $O(d^2)$ roots is then in $\sOB(d^8+d^7\tau)$.  Moreover, the
  algorithm requires testing $O(\log L)$ times whether some of the $O(d^2)$ boxes
  intersect, which can be done, in total, with $O(\log L)$ times $\sO(d^2)$
  arithmetic operations (see e.g. \cite[\S 3]{ze-fsbi-02}) 
  and thus with $\sOB(d^8+d^7\tau)$ bit operations since the vertices of the box
  vertices have bitsize in $\sO(d^6+d^5\tau)$.

  Therefore, we can compute isolating boxes for the solutions of $\ideal{P,Q}$
  in $\sOB(d^8+d^7\tau)$ bit operations, and the box vertices have bitsize in
  $\sOB(d^6+d^5\tau)$.

\medskip\noindent\emph{Bitsize of the box vertices.} \quad We finally show how
to compute, from the isolated boxes with vertices of bitsize in
$\sO(d^6+d^5\tau)$, some larger isolating boxes whose vertices have bitsize in
$\sO(d^3\tau)$. The method is identical for the $X$ or the $Y$-coordinates of
the boxes, thus we only consider the $x$-coordinates.  {We iteratively
  refine the boxes as describe above except that, once none of the boxes
  intersect, we carry on with the iterative refinement of the boxes until the
  distance in $X$ between any two boxes that do not overlap in $X$ is larger
  than $\frac12 2^{-\varepsilon}$ where $\varepsilon$, as defined at the
  beginning of the proof,
  is such that the distance between any two roots of the
  resultant of $P$ and $Q$ with respect to $X$ is at least $2^{-\varepsilon}$;
  we use here an explicit value for $\varepsilon$ which is given by
  Lemma~\ref{lem:sepbound}. 
  On the other hand, if we were to refine all the boxes until their width are
  less than $2^{-\varepsilon'}=\frac14 2^{-\varepsilon}$, the distance between
  any two boxes that do not overlap in $X$ would be ensured to be larger than
  $\frac12 2^{-\varepsilon}$. Hence the above analysis of the algorithm still
  applies since we considered that all boxes could be refined until their width
  (and height) do not exceed $2^{-\varepsilon'}$.  }

Now, for every box, all the other boxes that do not overlap in $X$ are at
distance more than $\frac12 2^{-\varepsilon}$ in $X$ (before enlargement), so
the considered box can be enlarged in $X$ using coordinates in intervals of
length at least $\frac{1}{4}2^{-\varepsilon}$ on the left and on the right sides
of the box. We conclude the argument by noting that, given any such interval
$[a,b]$ of width at least $2^{-\varepsilon'}$ with $\varepsilon'=\varepsilon+2
\in\sO(d^3\tau)$ and such that $|a|$ and $|b|$ are smaller than $2^{\sigma}$
with $\sigma=\sO(d^2+d\tau)$ (by Cauchy bound, as noted above), we can easily
compute in that interval a rational of bitsize at most
$\varepsilon'+\sigma\in\sO(d^3\tau)$.\footnote{A rational of bitsize at most
  $\varepsilon'+\sigma$ can be constructed as follows. We can assume without
  loss of generality that $a$ and $b$ are both positive since the case where
  they are both negative is symmetric and, otherwise, the problem is
  trivial. Let $q_k$ be the truncation of $b$ after the $k$-th digits of the
  mantissa, i.e.  $q_k=\lfloor b2^k \rfloor 2^{-k}$, and let $k_1$ be the
  smallest nonnegative integer such that $q_{k_1} \geq a$. By construction
  $q_{k_1} \in [a,b]$ and we prove that its bitsize is at most
  $\varepsilon'+\sigma$.
  If $k_1=0$, $q_{k_1} =\lfloor b \rfloor\leq 2^\sigma$ thus $q_{k_1}$ has
  bitsize at most $\sigma$. Otherwise, with $k_0=k_1-1$, we have $q_{k_0} < a$
  which implies that $b-q_{k_0} > b-a \geq 2^{-\varepsilon'}$. On the other
  hand, $b-q_{k_0} = 2^{-k_0}(b2^{k_0}-\lfloor b2^{k_0} \rfloor) < 2^{-k_0}$,
  thus $2^{-\varepsilon'} <2^{-k_0} $ and $ \varepsilon' > k_0 $. It follows
  that the bitsize of $q_{k_1}$, which is $k_1$ plus the bitsize of $\lfloor b
  \rfloor$, is less than $\varepsilon'+1$ plus $\sigma$.}
\end{proof}

{
\begin{remark}
It is straightforward that the above  proof and proposition also hold if a parameterization of
Gonzalez-Vega and El Kahoui \cite{VegKah:curve2d:96} is given  instead of a RUR.
\end{remark}
}
\subsection{{Sign of a polynomial at the  solutions of a system}}\label{sec:sign_at}
%%%%%%%%%%%%%%%%%%%%%%%%%%%%%%%%%%%%%%%%%%%%%%%%%

This section addresses the problem of computing the sign ($+, -$ or 0) of a
given polynomial $F$ at the solutions of a bivariate system defined by two
polynomials $P$ and $Q$.  We consider in the following that all input
polynomials, $P$, $Q$ and $F$ are in $\Z[X,Y]$, have degree at most $d$ and
coefficients of bitsize at most $\tau$. We assume without loss of generality
that the bound $d$ is \emph{even}. Recall that, as mentioned in the
  introduction, the best known complexity for this problem is to our knowledge
  $\sOB(d^{10}+d^9\tau)$ for the sign at one real solution and
  $\sOB(d^{12}+d^{11}\tau)$ for the sign at all the solutions (see \cite[Th. 14
  \& Cor. 24]{det-jsc-2009} with the improvement of
  \cite{sagraloff2012issacNewDsc} for the root isolation).  We first describe a
  naive  RUR-based {\em sign\_at} algorithm for computing the sign at one real solution of
  the system, which runs in $\sOB(d^9+d^8\tau)$ time. Then, using properties of
  generalized Sturm sequences, we
  analyze 
  a more efficient algorithm that runs in $\sOB(d^8+d^7\tau)$ time. We also show
that  the sign of $F$  at the $O(d^2)$ solutions of the system can be computed in  only $O(d)$ times that for one real solution.

Once the RUR $\{f_{I,a},f_{I,a,1},f_{I,a,X},f_{I,a,Y}\}$ of $I=\ideal{P,Q}$ is
computed, we can use it to translate a bivariate sign computation into a
univariate sign computation.  Indeed, let $F(X,Y)$ be the polynomial to be
evaluated at the solution $(\alpha,\beta)$ of $I$ that is the image of the root
$\gamma$ of $f_{I,a}$ by the RUR mapping.  We first define the polynomial $f_F(T)$
roughly as the numerator of the rational fraction obtained by substituting
$X=\frac{f_{I, a, X} (T)}{f_{I, a, 1} (T)}$ and $Y=\frac{f_{I, a, Y} (T)}{f_{I,
    a, 1} (T)}$ in the polynomial $F(X,Y)$, so that the sign of
$F(\alpha,\beta)$ is the same as that of $f_F(\gamma)$.

%%%%%%%%%%%%
\newcommand{\dF}{d}
%%%%%%%%%%%%
\begin{lemma}\label{lem:f_F}
  The primitive part\footnote{{See definition in
      Section~\ref{sec:size-rur}.}} of $f_F(T)=f_{I,a,1}^{\dF}(T)F(T-aY,Y)$,
  with $Y=\frac{f_{I,a,Y}(T)}{f_{I,a,1}(T)}$, has degree $O(d^3)$, bitsize in
  $\sO(d^3+d^2\tau)$, and it can be computed with $\sO_B(d^7+d^6\tau)$ bit
  operations.
  The sign of $F$ at a real solution of $I=\ideal{P,Q}$ is equal to the sign of
  $pp(f_F)$ at the corresponding root of $f_{I,a}$ via the mapping of the RUR.
\end{lemma}

\begin{proof}
  We first compute the polynomial $F(T-aY,Y)$ in the form $\sum_{i = 0}^{\dF}
  a_i (T) Y^i$. Then, $f_F(T)$ is equal to $\sum_{i = 0}^{\dF} a_i (T) f_{I, a,
    Y} (T)^i f_{I, a, 1} (T)^{\dF - i}$. Consequently, computing an expanded
  form of $f_F(T)$ can be done by computing the $a_i(T)$, the powers $f_{I, a,
    Y}(T)^i$ and $f_{I, a, 1}(T)^i$, and their appropriate products and sum.

  \medskip\noindent\emph{Computing $a_i(T)$.} \quad According to
  Lemma~\ref{lem:complexity:shear}, $P(T-SY,Y)$ can be expanded with
  $\sOB(d^4+d^3\tau)$ bit operations and its bitsize is in $\sO(d+\tau)$. These
  bounds also apply to $F(T-SY,Y)$ and we deduce $F(T-aY,Y)$ by substituting $S$
  by $a$. Writing $F(T-SY,Y)=\sum_{i=0}^{\dF}{f_i(T,Y)S^i}$, the computation of
  $F(T-aY,Y)$ can be done by computing and summing the $f_i(T,Y)a^i$.  Since $a$
  has bitsize in $O(\log d)$ by hypothesis, $a ^i$ has bitsize in $O(\dF\log
  d)\subseteq \sO(d)$, and computing all the $a^i$ can be done with $\sO_B(d^2)$
  bit operations. For each $a^i$, computing $f_i(T,Y)a^i$ can be done with
  $O(d^2)$ multiplications between integers of bitsize in $\sO(d+\tau)$, and
  thus with $\sOB(d^2(d+\tau))$ bit operations.  Thus, computing all the
  $f_i(T,Y)a^i$ can be done with $\sOB(d^3(d+\tau))$ bit operations, and
  summing, for every of the $O(\dF^2)$ monomials in $(T,Y)$, $\dF$ coefficients
  (corresponding to every $i$) of bitsize in $\sO(d+\tau)$ can also be done with
  $\sOB(d^3(d+\tau))$ bit operations, in total.  It follows that, $F(T-aY,Y)$
  and thus all the $a_i(T)$ can be computed with $\sOB(d^4+d^3\tau)$ bit
  operations.

  \medskip\noindent\emph{Computing $f_{I, a, Y}(T)^i$ and $f_{I, a, 1}(T)^i$.}  \quad $f_{I, a,
    Y}(T)$ has degree $\OO(d^2)$ and bitsize $\sO(d^2+d\tau)$ (by Theorem~\ref{th:rur}), thus $f_{I,
    a, Y}(T)^i$ has degree in $\OO(d^3)$ and bitsize in $\sO(d^3+d^2\tau)$. Computing all the $f_{I,
    a, Y}(T)^i$ can be done with $\OO(\dF)$ multiplications between these polynomials.  Every
  multiplication can be done with a bit complexity that is soft linear in the product of the maximum
  degrees and maximum bitsizes \cite[Corollary 8.27]{vzGGer}, thus all the multiplications can be
  done with $\sOB(d^4(d^3+d^2\tau))$ bit operations in total.  It follows that all the $f_{I, a,
    Y}(T)^i$, and similarly all the $f_{I, a,1}(T)^i$, can be computed using $\sOB(d^7+d^6\tau)$ bit
  operations and their bitsize is in $\sO(d^3+d^2\tau)$.

  \medskip\noindent\emph{Computing $f_F(T)$.}  \quad Computing $a_i (T) f_{I, a, Y}(T)^i f_{I, a,
    1}(T)^{\dF - i}$, for $i=0,\ldots,\dF$, amounts to multiplying $\OO(\dF)$ times, univariate
  polynomials of degree $\OO(d^3)$ and bitsize $\sO(d^3+d^2\tau)$, which can be done, similarly as
  above, with $\sO(d^7+d^6\tau)$ bit operations. Finally, their sum is the sum of $\dF$ univariate
  polynomials of degree $\OO(d^3)$ and bitsize $\sO(d^3+d^2\tau)$, which can also be computed within
  the same bit complexity. Hence, $f_F(T)$ can be computed with $\sOB(d^7+d^6\tau)$ bit operations
  and its coefficients have bitsize in $\sO(d^3+d^2\tau)$.

  \medskip\noindent\emph{Primitive part of $f_F(T)$.} \quad
  According to Proposition \ref{prop:rur-size}, there exists an integer $r$ of
  bitsize in $\sO(d^2+d\tau)$ such that its product with the RUR polynomials
  gives polynomials in $\Z[T]$ of bitsize in $\sO(d^2+d\tau)$. Consider the
  polynomial $r^df_F(T)=(rf_{I,a,1}(T))^dF(T-aY,Y)$ with
  $Y=\frac{rf_{I,a,Y}(T)}{rf_{I,a,1}(T)}$. This polynomial has its coefficients
  in $\Z$ since $rf_{I,a,Y}(T)$ and $rf_{I,a,1}(T)$ are in $\Z[T]$. Moreover,
  since $rf_{I,a,Y}(T)$ and $rf_{I,a,1}(T)$ have bitsize in $\sO(d^2+d\tau)$,
   $r^df_F(T)$ can be computed, similarly as above,  in $\sOB(d^7+d^6\tau)$
  and it has bitsize in $\sO(d^3+d^2\tau)$. 
The primitive  part of $f_F(T)$ has also bitsize  in $\sO(d^3+d^2\tau)$ (since it is smaller
than or equal to that of $r^df_F(T)$) and it can be computed from $r^df_F(T)$ with $\sOB(d^3(d^3+d^2\tau))$ bit operations
by computing $O(d^3)$ $\gcd$ of coefficients of bitsize $\sO(d^3+d^2\tau)$ \cite[\S 2.A.6]{Yap-2000}.

  \medskip\noindent\emph{Signs of $F$ and $f_F$.} \quad It remains to show that the sign of $F$ at a
  real solution of $I=\ideal{P,Q}$ is the sign of $f_F$ at the corresponding root of $f_{I,a}$ via
  the mapping of the RUR.  By Definition \ref{def:rur}, there is a one-to-one mapping between the
  roots of $f_{I,a}$ and those of $I=\ideal{P,Q}$ that maps a root $\gamma$ of $f_{I,a}$ to a
  solution
  $(\alpha,\beta)=(\frac{f_{I,a,X}(\gamma)}{f_{I,a,1}(\gamma)},\frac{f_{I,a,Y}(\gamma)}{f_{I,a,1}(\gamma)})$
  of $I$ such that $\gamma=\alpha+a\beta$ and $f_{I,a,1}(\gamma)\neq 0$.  For any such pair of
  $\gamma$ and $(\alpha,\beta)$,
  $f_F(\gamma)=f_{I,a,1}^{\dF}(\gamma)F(\gamma-a\frac{f_{I,a,Y}(\gamma)}{f_{I,a,1}(\gamma)},\frac{f_{I,a,Y}(\gamma)}{f_{I,a,1}(\gamma)})$
  by definition of $f_F(T)$, and thus $f_F(\gamma)=f_{I,a,1}^{\dF}(\gamma)F(\alpha,\beta)$. It
  follows that $f_F(\gamma)$ and $F(\alpha,\beta)$ have the same sign since $f_{I,a,1}(\gamma)\neq
  0$ and $\dF$ is even by hypothesis.
\end{proof}

\paragraph{Naive algorithm.}
The knowledge of a RUR $\{f_{I,a},f_{I,a,1},f_{I,a,X},f_{I,a,Y}\}$ of
$I=\ideal{P,Q}$ yields a straightforward algorithm for computing the
sign of $F$ at a real solution of $I$. Indeed, it is sufficient to
isolate the real roots of $f_{I,a}$, so that the intervals are also
isolating for $f_{I,a}f_F$, and then to evaluate the sign of $\overline{f_F}$ at
the endpoints of these isolating intervals. We analyze the complexity
of this straightforward algorithm before describing our more subtle
and more efficient algorithm. We provide this analysis for several
reasons: first it answers a natural question, second it shows that
even a {RUR-based} naive algorithm performs better than the state of the art.

\begin{lemma}\label{lem:sign_at_rur1}
Given a RUR $\{f_{I,a},f_{I,a,1},f_{I,a,X},f_{I,a,Y}\}$ of $I=\ideal{P,Q}$ (satisfying
the bounds of Theorem~\ref{th:rur})  and an isolating
interval for   a real root $\gamma$ of $f_{I,a}$,  the sign of $F$ at the real solution of $I$
that corresponds to $\gamma$ can be computed with $\sO_B(d^9+ d^8\tau)$ bit operations.
\end{lemma}
\begin{proof}
  By Lemma~\ref{lem:f_F}, $pp(f_F)$ has degree  $O(d^3)$ and bitsize
   $\sO(d^3+d^2\tau)$, and it can be computed with $\sO_B(d^7
  +d^6\tau)$ bit operations. By Theorem~\ref{th:rur}, $f_{I, a}$ has
  degree  $O(d^2)$ and bitsize  $\sO(d^2+d\tau)$, thus the product
  $pp(f_F)\,f_{I, a}$ has degree  $\OO(d^3)$ and bitsize 
  $\sO(d^3+d^2\tau)$. By Lemma~\ref{lem:sepbound}, the root separation
  bound of $pp(f_F)\,f_{I, a}$ has bitsize $\sO(d^6+d^5\tau)$. We
  refine the isolating interval of $\gamma$ for $f_{I, a}$ to the
  precision of the root separation bound of $pp(f_F)\,f_{I, a}$,
  which can be done with
  $\sOB((d^2)^2(d^2+d\tau)+d^2(d^6+d^5\tau))=\sOB(d^8+d^7\tau)$ bit
  operations according to Lemma~\ref{comp:isolation}.  Furthermore, we
  can ensure that the new interval has rational endpoints with bitsize
   $\sO(d^6+d^5\tau)$, similarly as in the proof of
  Proposition~\ref{prop:computing-boxes}. 
  On the other hand, {by Lemma~\ref{complexity:gcd}, since $pp(f_F)$ has bitsize
   $\sO(d^3+d^2\tau)$, its squarefree part $\overline{pp(f_F)}$ can be computed in complexity $\sOB((d^3)^2(d^3+
  d^2\tau))=\sOB(d^9+ d^8\tau)$ and it has bitsize in $\sOB(d^3+ d^2\tau)$.}
  It then follows from Lemma~\ref{lem:comp:evaluation} that the
  evaluation of $\overline{pp(f_F)}$ at the boundaries of the
  refined interval can be done with $\sO_B(d^3(d^6+d^5\tau))$ bit
  operations which concludes the proof by Lemma~\ref{lem:f_F}.
\end{proof}

\paragraph{Improved algorithm.}
Our more subtle algorithm is, in essence the one presented by Diochnos et
al. for evaluating the sign of a univariate polynomial (here $pp(f_F)$) at the roots
of a squarefree univariate polynomial (here $\overline{f_{I,a}}$)
\cite[Corollary 5]{det-jsc-2009}. The idea of this algorithm comes originally
from \cite{lr-jsc-2001}, where the Cauchy index of two polynomials is computed
by means of sign variations of a particular remainder sequence called the
Sylvester-Habicht sequence. In \cite{det-jsc-2009}, this approach is slightly
adapted to deduce the sign from the Cauchy index (\cite[Theorem 7.3]{Yap-2000})
and the bit complexity is given in terms of the two initial degrees and
bitsizes. 
Unfortunately, the
  corresponding proof is problematic because the authors refer to two
  complexity results for computing parts of the Sylvester-Habicht sequences and
  none of them actually applies.\footnote{Precisely, their proof is based on their Proposition~1 which
  claims, based on \cite{lr-jsc-2001} and \cite{Reischert1997} that given two
  polynomials $f$ and $g$ of degree $p>q$ and bitsize in $O(\tau)$, any of their
  polynomial subresultants as well as the whole quotient chain corresponding to
  the subresultant sequence can be computed with $\sOB(pq\tau)$ bit
  operations. However, in \cite{lr-jsc-2001} the complexity results are not
  stated in terms of $p$ and $q$ but only in terms of the maximum degree while
  in \cite{Reischert1997}, the result assumes that the $(q-1)^{th}$ subresultant
  of $f$ and $g$ is known.}
Following the spirit of their approach, we present in Lemma~\ref{DET09-revisited} a new {(weaker)}
complexity result for evaluating the sign of a univariate polynomial at the roots
of a squarefree univariate polynomial. This result is used to derive the bit
complexity of evaluating the sign of a bivariate polynomial at the roots of the
system.  For clarity, we postpone the proof of this 
 lemma to
Section~\ref{sec:proof-sign_at} after Theorem~\ref{th:sign_at_rur}.

\begin{lemma}\label{DET09-revisited} Let $f\in\Z[X]$ be a squarefree polynomial of degree $d_f$ and bitsize
  $\tau_f$, and $(a,b)$
be an isolating interval of one of its real roots $\gamma$ with $a$ and $b$ distinct rationals of bitsize in $\sO(d_f\tau_f)$
and  $f(a)f(b) \neq~0$. Let $g\in\Z[X]$ be of degree $d_g$ and bitsize $\tau_g$.
The sign of $g(\gamma)$ can be computed in  $\sOB((d_f^3+d_g^2)\tau_f + (d_f^2+d_fd_g)\tau_g)$ 
bit operations.
The sign of $g$ at all the {real} roots of $f$ can be computed with
$\sOB((d_f^3+d_f^2d_g+d_g^2)\tau_f + (d_f^3+d_fd_g)\tau_g)$ bit operations.
\end{lemma}

\begin{theorem}
\label{th:sign_at_rur}
Given a RUR $\{f_{I,a},f_{I,a,1},f_{I,a,X},f_{I,a,Y}\}$ of $I=\ideal{P,Q}$
(satisfying the bounds of Theorem~\ref{th:rur}), the sign of $F$ at a real
solution of $I$ can be computed with $\sO_B(d^8+ d^7\tau)$ bit operations. The sign of $F$ at all
the solutions of $I$ can be computed with $\sOB(d^9+d^8\tau)$ bit operations.
\end{theorem}
\begin{proof}
By Lemma~\ref{lem:f_F}, the sign of $F$ at the real solutions of $I$,
    is equal to the sign of $pp(f_F)$ at the corresponding roots of $f_{I,a}$,
    or equivalently at those of $\overline{pp(f_{I,a})}$.  {Furthermore, $pp(f_F)$ 
has degree $O(d^3)$, bitsize in
  $\sO(d^3+d^2\tau)$, and it can be computed with $\sO_B(d^7+d^6\tau)$ bit
  operations. On the other hand, } by
      Theorem~\ref{th:rur} and Proposition~\ref{prop:rur-size}, the primitive
      part of $f_{I,a}$ has degree at most $d^2$ and bitsize in
      $\sO(d^2+d\tau)$. Since $f_{I,a}$ is monic (see Equation~\eqref{eq:defRUR}), its
      primitive part can be computed by multiplying it by the lcm of the
      denominators of its coefficients. This lcm
can be computed with $O(d^2)$ lcms of integers whose bitsizes remain in
$\sO(d^2+d\tau)$ (since $f_{I,a}$ is monic and its primitive part has bitsize in
$\sO(d^2+d\tau)$). Each lcm can be computed with $\sOB(d^2+d\tau)$ bit
operations \cite[\S 2.A.6]{Yap-2000}, thus $pp(f_{I,a})$ can be computed in
$\sO(d^4+d^3\tau)$ bit operations.\footnote{Note that is if $f_{I,a}$ has been
  computed using Proposition~\ref{prop:rur-res2}, then instead of computing
  $pp(f_{I,a})$ one can consider $R(T,a)=f_{I,a}(T)\,L_R(a)$ which is a
  polynomial of degree $O(d^2)$ with integer coefficients of bitsize
  $\sO(d^2+d\tau)$ by Lemma~\ref{lem:complexity:shear}.}  The squarefree part of
$pp(f_{I,a})$ can thus be computed in $\sOB(d^4(d^2+d\tau))$ bit operations and
it has bitsize in $\sO(d^2+d\tau)$, by Lemma~\ref{complexity:gcd}.  By
Lemmas~\ref{comp:isolation} and \ref{lem:sepbound}, the isolating intervals (if
not given) of $\overline{pp(f_{I,a})}$ can be computed in
$\sOB((d^2)^3(d^2+d\tau))$ bit operations with intervals boundaries of bitsize
satisfying the hypotheses of Lemma~\ref{DET09-revisited}. Indeed, we can
  ensure during the isolation of the roots of $f=\overline{pp(f_{I,a})}$ that
  the isolating intervals have endpoints with bitsize in $\sO(d_f\tau_f)$,
  similarly as in the proof of Proposition~\ref{prop:computing-boxes}. 
Applying Lemma \ref{DET09-revisited} then concludes the proof. 
\end{proof}

{
\begin{remark}
Theorem~\ref{th:sign_at_rur} holds also if the solutions of $I=\ideal{P,Q}$ are described by the rational parameterization of  
   Gonzalez-Vega and El Kahoui \cite{VegKah:curve2d:96} instead of a RUR. Indeed, such parameterization is defined,
   in the worst case, by $\Theta(d)$  univariate polynomials $f_i$ of degree $d_{f_i}$ whose sum
   $d_f$ is at most
   $d^2$, and by  associated rational one-to-one mappings which are defined, as for the RUR, by
   polynomials of degree $O(d^2)$ and bitsize $O(d^2+d\tau)$. The result of
   Theorem~\ref{th:sign_at_rur} on the sign of $F$ at \emph{one} real solution of $I$ thus trivially
   still holds. For the sign of $F$ at \emph{all} real solutions of $I$ the result also holds from
   the following    observation. In the proofs of Lemmas~\ref{lem:lick-roy} and \ref{DET09-revisited}, the
   computation of one sequence of  unevaluated Sylvester-Habicht transition matrices has complexity
   $\sOB(pH)$ (in proof of Lemma~\ref{lem:lick-roy}) where $p$  is in $O(d_{f_i}+d_g)$ in the proof
   of Lemma~\ref{DET09-revisited}. The sum of the $pH$ over all $i$ is thus $O((d_{f}+dd_g)H)$
   instead of $O((d_{f}+d_g)H)$ as for the RUR. However, $d_gH$ writes in the proof
   of Lemma~\ref{DET09-revisited} as
   $\sO(d_g((d_f+d_g)\tau_f+d_f(\tau_f+\tau_g)))=\sO(d_fd_g(\tau_f+\tau_g)+d_g^2\tau_f)$ which
   writes in the proof of Theorem~\ref{th:sign_at_rur} as
   $\sO(d^2d^3(d^3+d^2\tau)+(d^3)^2(d^2+d\tau))=\sO(d^8+d^7\tau)$. Thus multiplying this by $d$
   remains within the targeted bit complexity. On the other hand, the complexity of the evaluation
   phase in the proofs of Lemmas~\ref{lem:lick-roy} and \ref{DET09-revisited} does not increase when
   considering the representation of Gonzalez-Vega and El Kahoui instead of the RUR because the
   total complexity of the  evaluations depends only on the number of solutions at
   which we evaluate the sign of the other polynomial and on the degree and bitsize of the
   polynomials involved, and both of them are in the same complexity in both representations (only
   the number of polynomials is larger in Gonzalez-Vega and El Kahoui representation). 
\end{remark}
}

\subsubsection{Proof of Lemma~\ref{DET09-revisited}}\label{sec:proof-sign_at}

{As shown in \cite[Theorem 2.61]{BPR06}, the sign of $g(\gamma)$ is $
  V(SRemS(f,f'g;a,b))$ where $V(SRemS(P,$ $ Q;a,b))$} is the number of sign
variations in the signed remainder sequence of $P$ and $Q$ evaluated at $a$
minus the number of sign variations in this sequence evaluated at $b$
{(see Definition 1.7 in \cite{BPR06} for the sequence and Notation 2.32
  for the sign variation).}  {On the other hand, for any $P$ and $Q$
  such that $\deg(P)>\deg(Q)$ and $P(a)\,P(b) \neq 0$ or $Q(a)\,Q(b) \neq 0$, we
  have according to \cite[Theorems 3.2, 3.18 \& Remarks 3.9,
  3.25]{Roy96}\footnote{The same result can be found directly stated,
      in French, in \cite[Theorem 4]{Lombardi90}.}  that
  $V(SRemS(P,Q;a,b))=W(SylH(P,Q;a,b))$ where $SylH$ is the Sylvester-Habicht
  sequence of $P$ and $Q$, and $W$ is the related sign variation
  function.}\footnote{{The Sylvester-Habicht sequence, defined in
    \cite[\S 8.3.2.2]{BPR06} as the Signed Subresultant sequence, can be derived
    from the classical subresultant sequence \cite{Kahoui03}
by   multiplying the two starting subresultants by $+1$ the next two by $-1$ and so
  on.  $W$ is defined as the usual sign variation with the following
  modification for groups of two consecutive zeros: count \emph{one}
  sign variation for the groups $[+,0,0,-]$ and $[-,0,0,+]$, and \emph{two} sign
  variations for the $[+,0,0,+]$ and $[-,0,0,-]$ (see \cite[\S 9.1.3 Notation 9.11]{BPR06}).}}
The following intermediate result  is a consequence of an adaptation of
\cite[Theorem 5.2]{lr-jsc-2001} in the case
where the polynomials $P$ and $Q$ have
different degrees and bitsizes.

\begin{lemma}\label{lem:lick-roy}
  Let $P$ and $Q$ in $\Z[X]$ with $deg(P)=p>q=deg(Q)$ and bitsize respectively
  $\tau_P,\tau_Q$. If $a$ and $b$ are  two rational numbers of bitsize bounded by
  $\sigma$,  
  the computation of $W(SylH (P,Q;a,b))$ can be performed with
 $\sOB((p+q^2)\sigma+p(p\tau_Q+q\tau_P))$ bit operations. 

Moreover, if $a_\ell$ and $b_\ell$, $1\leq \ell\leq u$, are rational numbers of bitsizes that sum to 
  $\sigma$, 
  the computation of $W(SylH (P,Q;a_\ell,b_\ell))$ can be performed  for all $\ell$ with
$\sOB((p+q^2)\sigma+(p+qu)(p\tau_Q+q\tau_P)+pu\tau_P)$ bit operations. 
\end{lemma}

\begin{proof} 
  Following the algorithm in \cite{lr-jsc-2001}, we first compute the
  consecutive Sylvester-Habicht transition matrices of $P$ and $Q$ denoted by
  $\mathcal{N}_{j,i}$ with $ 0 \leq j < i \leq p$. These matrices link 
  consecutive regular couples\footnote{Regular couples in the
    Sylvester-Habicht sequence are the nonzero Sylvester-Habicht
    polynomials $(Sh_i, Sh_{i-1})$ such that $\deg(Sh_i) > \deg(Sh_{i-1})$.} $(Sh_i,Sh_{i-1})$ and $(Sh_j,Sh_{j-1})$ in
  the Sylvester-Habicht sequence as
  follows:
  \begin{equation}\label{eq:trans-mat}
    \begin{pmatrix} Sh_{j} \\ Sh_{j-1} \end{pmatrix} = \mathcal{N}_{j,i}
    \begin{pmatrix} Sh_{i} \\ Sh_{i-1} \end{pmatrix} 
    \ such \ that \ i \leq p \ and \ \ (Sh_{p}, Sh_{p-1})=(P,Q). 
  \end{equation}
  According to \cite[Theorem. 5.2 \& Corollary 5.2]{lr-jsc-2001}, computing all the
  matrices $\mathcal{N}_{j,i}$ of $P$ and $Q$ can be done
  with $\sOB(pH)$ bit operations, where $H\in \sO(q\tau_P+p\tau_Q)$ is an upper bound
 on the bitsize appearing in the computations given by Hadamard's inequality.

  We  evaluate the Sylvester-Habicht sequence at a rational $a$ by first evaluating   $P$, $Q$,
  and all 
  the  matrices $\mathcal{N}_{j,i}$ at $a$,  and then by applying
  iteratively the above formula. Doing the  same at $b$  yields $W(SylH (P,Q;a,b))$.

  First, note that the evaluation of $P(a)$ and $Q(a)$ can be done with $\sOB(p(\tau_P+\sigma))$ plus
  $\sOB(q(\tau_Q+\sigma))$, that is   $\sOB(p(\tau_P+\tau_Q+\sigma))$ bit operations (since $p>q$), by Lemma~\ref{lem:comp:evaluation}.
 The polynomials appearing in the 
  matrices $\mathcal{N}_{j,i}$ have bitsize at most $H$
  and the sum of their degrees is equal to
  $p$ \cite[Corollary 4.3]{lr-jsc-2001}.\footnote{\cite[Corollary 4.3]{lr-jsc-2001} states that consecutive
    Sylvester-Habicht transition matrices consist of one zero, two integers and a polynomial which
    is, up to a coefficient, the quotient of the division of two consecutive Sylvester-Habicht
    polynomials. These polynomials being proportional to polynomials in the remainder sequence of
    $(P,Q)$, the sum of the degrees of their quotients is equal to the degree of $P$.}
    Thus, all $\mathcal{N}_{j,i}(a)$ have
    bitsize $\sO(p\sigma+H)$ and they can
  be computed in a total of $\sOB(p(\sigma+H))$ bit operations, by Lemma~\ref{lem:comp:evaluation}.  
Moreover, by considering the 
matrices $\mathcal{N}_{j,i}$ other than the first one  $\mathcal{N}_{k,p}$, as
the consecutive transition matrices of the Sylvester-Habicht sequence of the first regular couple
  $(Sh_k,Sh_{k-1})$ after $(Sh_p,Sh_{p-1})$, we have that the polynomials appearing in these
  matrices have the sum of their degrees equal to that of $Sh_k$ 
  which is at most $q$ (since $k\leq p-1$ and $Sh_{p-1}=Q$). Thus, except the first one $\mathcal{N}_{k,p}(a)$,  all evaluated 
  matrices $\mathcal{N}_{j,i}(a)$ have bitsize $\sO(q\sigma+H)$ and they can be computed in a total of  $\sOB(q(\sigma+H))$ bit operations.

We now apply iteratively Equation~\eqref{eq:trans-mat}  for
  computing all the $Sh_i(a)$. Since all Sylvester-Habicht polynomials have bitsize at most $H$ and degree at most $q$ except
  the first one $Sh_p=P$, the bitsize of $Sh_{i<p}(a)$ is in $O(q\sigma+H)$ and that
  of $Sh_{p}(a)$ is in $O(p\sigma+\tau_P)$.
Given $P(a), Q(a)$ and all $\mathcal{N}_{j,i}(a)$, it follows from their  bitsizes that we can
compute iteratively the $Sh_{i}(a)$
in time $\sOB(p\sigma+H)$ for the first regular couple after
$(Sh_p,Sh_{p-1})=(P,Q)$ and in time $\sOB(q\sigma+H)$ for each of the
others. Thus, for computing of $W(SylH (P,Q;a,b))$,
the initial computation of all $\mathcal{N}_{j,i}$ takes $\sOB(pH)$ bit
operations and
the evaluation phase takes $\sOB(p(\tau_P+\tau_Q+\sigma))$ plus
$\sOB(p(\sigma+H)+q(q\sigma+H))$ bit operations, which gives a total of  $\sOB(p(\sigma+H) + q^2\sigma)$ bit operations.

We now consider the case of computing $W(SylH (P,Q;a_\ell,b_\ell))$  for $1\leq \ell\leq u$.
We  slightly change the above algorithm as follows. We only change the way to evaluate
the first regular couple $(Sh_k,Sh_{k-1})$ after $(Sh_p,Sh_{p-1})$ 
at the $a_\ell$ (and $b_\ell$). Once the matrices $\mathcal{N}_{j,i}$ have been computed, we compute the (non-evaluated) first
regular couple $(Sh_k,Sh_{k-1})=\mathcal{N}_{k,p}(Sh_p,Sh_{p-1})$. Since the polynomials in
$\mathcal{N}_{k,p}$ have degree at most $p$ and  bitsize at most $H$, the couple $(Sh_k,Sh_{k-1})$
can be computed in $\sOB(p(H+\tau_P+\tau_Q))=\sOB(pH)$ time \cite[Corollary 8.27]{vzGGer}. 
As noted above, $Sh_k$, and thus also $Sh_{k-1}$,
have degree at most $q$ and they have bitsize at
most $H$, so  they can  be
evaluated at a given $a_\ell$ in time $\sOB(q(\sigma_\ell+H))$ where $\sigma_\ell$ is the bitsize of
$a_\ell$. Now, the  polynomials appearing in the matrices $\mathcal{N}_{j,i}$, other than
the first one $\mathcal{N}_{k,p}$, have bitsize at most $H$
and the sum of their degrees is at most $q$, so similarly as above, all the
$\mathcal{N}_{j,i}(a_\ell)$, 
except $\mathcal{N}_{k,p}(a_\ell)$, can be computed in total bit complexity $\sOB(q(\sigma_\ell+H))$.
Then, we compute as above each of the other regular couples evaluated at $a_\ell$ in time
  $\sOB(q\sigma_\ell+H)$.  Hence, the initial  computation of  all $\mathcal{N}_{j,i}$ and of
  $(Sh_k,Sh_{k-1})$ takes $\sOB(pH)$ bit operations and  the evaluation phase at all the $a_\ell$ takes
  the sum over $\ell$, $1\leq \ell\leq u$, of $\sOB(p(\tau_P+\tau_Q+\sigma_\ell))$  plus $\sOB(q(\sigma_\ell+H)+q(q\sigma_\ell+H))$
  bit operations, that is
$\sOB(p(\tau_P+\tau_Q) +(p+q^2)\sigma_\ell+ qH)$ which sums to $\sOB(pu(\tau_P+\tau_Q) +(p+q^2)\sigma+ quH)$.
Hence the total bit complexity for computing all the $W(SylH (P,Q;a_\ell,b_\ell))$  for $1\leq \ell\leq u$ is
$\sOB((p+q^2)\sigma+(p+qu)H+pu\tau_P)$ which concludes the proof.
\end{proof}

\begin{proof}[Proof of Lemma \ref{DET09-revisited}]
We may assume that $g$ has degree
  greater than one since, if $g$ is a constant the problem is trivial and, if $g(X)=cX-d$,
then the sign of $g(\gamma)$ follows from  (i)  the sign of $c$ if $\frac{d}{c} \not\in (a,b)$ 
and from (ii) the signs of $c$, $f(a)$, and  $f(\frac{d}{c})$ if $\frac{d}{c} \in ( a,b)$; indeed, the signs of $f(a)\neq 0$ and  $f(\frac{d}{c})$ determine whether $\gamma$ lies in
$(a,\frac{d}{c})$, $\{\frac{d}{c}\}$, or $(\frac{d}{c},b)$.
Hence, when $g$ has degree one, the sign of $g(\gamma)$ can be
computed with $\sOB(d_f(\tau_g+d_f\tau_f))$ bit operations according to
Lemma~\ref{lem:comp:evaluation}.

Recall that the sign of $g(\gamma)$ is $ V(SRemS(f,f'g;a,b))$  \cite[Theorem 2.61]{BPR06}.
  When $g$ has degree greater than one, we cannot directly apply
  Lemma~\ref{lem:lick-roy} since $\deg(f)<\deg(f'g)$. However, knowing the sign of
  $f$ and $f'g$ at $a$ and $b$ and noticing that their signed remainder sequence
  starts with $[f,f'g,-f,-rem(f'g,-f),\ldots]$, we can easily compute  the value 
  $c$ such that $V(SRemS(f,f'g;a,b))=V(SRemS(f'g,-f;a,b))+c$.  
  Furthermore, as observed at the beginning of this section and since
  $f(a)f(b)\neq 0$ by hypothesis, $V(SRemS(f'g,-f;a,b))=W(SylH(f'g,-f;a,b))$.
We can now apply Lemma \ref{lem:lick-roy} which thus yields the sign of $g(\gamma)$  with a bit
complexity in  $\sOB((p+q^2)\sigma+p(p\tau_Q+q\tau_P))$ which simplifies  into  $\sOB((d_f^3+d_g^2)\tau_f + (d_f^2+d_fd_g)\tau_g)$.

For the sign
  of $g$ at all the real roots of $f$, isolating intervals of these roots can be computed in
  complexity $\sOB(d_f^3\tau_f)$ (see Lemma~\ref{comp:isolation}) such that  the bitsizes of the
  interval boundaries sum up to $\sO(d_f^2 + d_f\tau_f)$  (a consequence of
  Davenport-Mahler-Mignotte bound, see e.g. \cite[Lemma 6]{det-jsc-2009}).
Similarly as for one root, Lemma \ref{lem:lick-roy} then yields that the sign of $g$  at all the real roots
of $f$ can be computed with a bit
complexity in  $\sOB((p+q^2)\sigma+(p+qu)(p\tau_Q+q\tau_P)+pu\tau_P)$ which writes as
$\sOB((d_f+d_g+d_f^2)d_f\tau_f+(d_f+d_g+d_f^2)((d_f+d_g)\tau_f+d_f(\tau_f+\tau_g))+(d_f+d_g)d_f(\tau_g+\tau_f))$ and simplifies
into  $\sOB((d_f^3+d_f^2d_g+d_g^2)\tau_f + (d_f^3+d_fd_g)\tau_g)$ bit operations.
\end{proof}

\subsection{Over-constrained systems}
\label{sec:overcontraint}

So far, we focused on systems defined by exactly two coprime polynomials. We now
extend our results to compute rational parameterizations of zero-dimensional
systems defined with additional equality or inequality. 
Let $P, Q \in \Z[X,Y]$ be two coprime polynomials of total degree at most $d$
and maximum bitsize $\tau$. In this section, we assume given
$RUR_{I,a}=\{f_{I,a},f_{I,a,1},f_{I,a,X},f_{I,a,Y}\}$ the RUR of the ideal
$I=\ideal{P,Q}$ associated to the separating form $X+aY$, we also assume that
the polynomials of this RUR satisfy the bitsize bound of
Theorem~\ref{th:rur}. Given another polynomial $F \in \Z[X,Y]$,
we have seen in the previous section how to compute the sign of $F$ at the
solutions of $I$.  With a similar approach, we now explain how to split
$RUR_{I,a}$ according to whether $F$ vanishes or not at the solutions of $I$.

Let $F \in \Z[X,Y]$ be of total degree at most $d$ and maximum bitsize $\tau$.
Identifying the roots of ${f_{I,a}}$ with the solutions of the system $I$ via
the RUR, let $f_{F=0}$ (resp. $f_{F\neq 0}$) be the squarefree factor of
$f_{I,a}$ such that 
its roots are exactly the solutions of the system $I$ at which the polynomial
$F$ vanishes (resp. does not vanish).

\begin{lemma}
\label{lem:rursplit}  
Given $RUR_{I,a}$,
the bit complexity of computing $f_{F=0}$ (resp. $f_{F\neq 0}$) is in
$\sOB(d^8+d^7\tau)$ and these polynomials have bitsize in $\sO(d^2+d\tau)$.
\end{lemma}
\begin{proof}
  The polynomial $f_F$ (not to be confused with $f_{F=0}$ or $f_{F\neq 0}$), as defined in
  Lemma~\ref{lem:f_F},   has the same sign as
  $F$ at the {real solutions} 
 of the system $I$. {The same holds for complex solutions by considering the ``sign'' as zero or
 nonzero.} The roots of the squarefree polynomial
  $f_{F=0}=\gcd(\overline{f_{I,a}},f_F)$ thus are the $\alpha + a \beta $ with
  $(\alpha,\beta)$ solution of $I$ and $F \left( \alpha, \beta \right) = 0$. The
  polynomial $f_{F\neq 0}$ defined as the gcd-free part of $\overline{f_{I,a}}$
  with respect to $f_F$ is also squarefree and encodes the solutions such
  that $F \left( \alpha, \beta \right) \neq 0$.
  
  According to Lemma~\ref{lem:f_F} and the proof of
  Theorem~\ref{th:sign_at_rur}, the primitive part of $f_F$ and
  $\overline{f_{I,a}}$ can be computed in, respectively, $\sO_B(d^7+d^6\tau)$
  and $\sOB(d^4(d^2+d\tau))$ bit operations. {Moreover,} these integer polynomials have,
  respectively, bitsize $\sO(d^3+d^2\tau)$ and $\sO(d^2+d\tau)$ and degree
  $O(d^3)$ and $O(d^2)$. Thus, by Lemma~\ref{lem:finegcd}, their gcd and the
  gcd-free part of $\overline{f_{I,a}}$ with respect to $f_F$, i.e. $f_{F=0}$
  and $f_{F\neq0}$, can be computed with $\sOB(d^8+d^7\tau)$ bit operations and
  they have bitsize in $\sO(d^2+d\tau)$.
\end{proof}

For several equality or inequality constraints, iterating this splitting process
gives a parameterization of the corresponding set of constraints. It is worth
noticing that the set of polynomials $\{f_{F=0},f_{I,a,1},f_{I,a,X},f_{I,a,Y}\}$ 
defines a rational parameterization of the solutions of the ideal $\ideal{P,Q,F}$, but this
is not a RUR of this ideal (in the sense of Definition~\ref{def:rur}). 
First,
because multiplicities are lost in the splitting process and second because the
coordinate polynomials of the parameterization are still those of the ideal
$I$. Still, it is possible to compute a RUR of the radical of the corresponding
ideal {(and similarly for the ideal corresponding to $F\neq0$)}:

\begin{proposition}
\label{prop:overconstrained}
  Given $RUR_{I,a}$ and $F \in \Z[X,Y]$ of total degree at most $d$ and maximum
  bitsize $\tau$, the bit complexity of computing the RUR of the radical of the
  ideal $\ideal{P,Q,F}$ is in $\sOB(d^8+d^7\tau)$.
\end{proposition}
\begin{proof}
  Denote by $J$ the radical of the ideal $\ideal{P,Q,F}$. The polynomial
  $f_{F=0}$ computed in
 Lemma~\ref{lem:rursplit} is the first polynomial
  $f_{J,a}$ of $RUR_{J,a}$.  Indeed, it vanishes at the solutions of this ideal
  (with identification of the roots of ${f_{J,a}}$ with the solutions of the
  system $J$) and it is squarefree. Then Proposition~\ref{prop:rur-res2} yields
  that $f_{J,a,1}$ is the gcd-free part of $f'_{J,a}$ with respect to
  $f_{J,a}$. As in the proof of Theorem~\ref{th:sign_at_rur}, $pp(f_{J,a})$ can
  be computed in $\sOB(d^4+d^3\tau)$ and has bitsize in $\sO(d^2+d\tau)$. By
  Lemma~\ref{lem:finegcd}, applied to $pp(f_{J,a})$ and its derivative,
  $f_{J,a,1}$ can be computed in $\sOB(d^6+d^5\tau)$. 

  According to Definition~\ref{def:rur} of a RUR, the $X$-coordinates of the
  solutions of $J$ are given by the polynomial fraction
  $\frac{f_{J,a,X}}{f_{J,a,1}}$ at the roots of
  $f_{J,a}$. On the other hand, the solutions of $J$, seen as solutions of
  $I$, 
  have their $X$-coordinates defined by the polynomial fraction
  $\frac{f_{I,a,X}}{f_{I,a,1}}$. This thus implies
  that$f_{J,a,X}=f_{I,a,1}^{-1}{f_{I,a,X}}{f_{J,a,1}}$ modulo $f_{J,a}$. The
  computation of $f_{I,a,1}^{-1}$ together with the multiplication with other
  polynomials of the RUR has already been studied in the proof of
  Proposition~\ref{prop:computing-boxes}; this can be done in
  $\sOB(d^6+d^5\tau)$ {time} and gives a polynomial of degree $O(d^2)$ and
  bitsize $\sO(d^4+d^3\tau)$. It remains to compute the remainder of the
  division of this polynomial with $f_{J,a}$, 
which can be done in a soft bit complexity of
  the order of the square of the maximum degree times the maximum bitsize,
  i.e. $\sOB(d^8+d^7\tau)$ \cite[Theorem 9.6 and subsequent
  discussion]{vzGGer}. A similar computation gives the polynomial $f_{I,a,Y}$,
  hence the computation of $RUR_{J,a}$ can be done in $\sOB(d^8+d^7\tau)$ {bit
    operations}.
\end{proof}

\section{Conclusion}
%%%%%%%%%%%%%%%%%%%%%%%%%%%%%%%%%%%%%%%%%%%%%%%%%

{

We studied the problem of solving systems of bivariate polynomials with integer
  coefficients using
  Rational Univariate Representations.  We first showed that the polynomials of the RUR  of a
  system of two polynomials can be expressed by simple formulas which yield a new simple method for
  computing the RUR and also yield a new bound on the bitsize of these polynomials.  This new
  bound implies, in particular, that the total space complexity of such RURs is, in the worst case,
  $\Theta(d)$ smaller than the alternative rational parameterization introduced by Gonzalez-Vega and
  El Kahoui \cite{VegKah:curve2d:96}. Given a RUR, this new bound also yields some improvements on the complexity
  of computing  isolating boxes and  performing sign\_at
  evaluations. Furthermore,  these improvements also hold for the  rational parameterization of  
   Gonzalez-Vega and El Kahoui. We also addressed the problem of computing RURs of over-constrained systems.

The algorithm we presented  for computing a RUR is more of a theoretical  than a practical
interest. Indeed, the computation of the resultant $R(T,S)$ of trivariate polynomials in not very
efficient in practice.     One particular problem of interest is thus the design of  a practical
efficient algorithm for computing RURs of bivariate systems whose bit complexity is as close as
possible to the one presented here. 
  Our complexity analysis shows that  our new algorithm for computing a RUR is
  dominated by that of finding a separating form which is in
  $\sOB(d^{8}+d^7\tau)$ \cite{bouzidi2013SepElt}. However, in a Monte-Carlo probabilistic setting,
  one can choose a candidate separating form randomly. On the other hand, in a Las-Vegas probabilistic
  setting, it is also possible to choose a candidate separating form randomly, compute a RUR-candidate using
  multi-modular arithmetic and taking advantage of our new bound on its bitsize, and verify a
  posteriori using the RUR-candidate if the chosen candidate separating form is actually
  separating. Such approach is the topic of current research and we refer to \cite{bouzidi:2011:inria-00580431:1} for preliminary work on
  the subject. Another problem of interest is to generalize  our  bounds on the bitsize of RURs to higher dimensions.
}

\small 
\bibliographystyle{alpha}
\bibliography{paper_separating_element.bib}
\end{document}